\documentclass [twoside,12pt] {article}
\usepackage{setspace,graphicx,bm,fancyhdr}
\usepackage{amsmath,amsthm,amssymb,latexsym}
\usepackage{color}
\usepackage{graphicx}
\usepackage{soul}

\usepackage{authblk}

\newcommand{\vct}[1]{\mbox{$\stackrel{\rightarrow}{#1}$}}

\newcommand{\slas}[1]{\mbox{${{#1} \!\!\! /}$}}

\begin{document}
\title{\bf Finite quantum gravity in locally conformally flat spacetime through covariant spectral regularization} 
\author{John Mashford \\
School of Mathematics and Statistics \\
University of Melbourne, Victoria 3010, Australia \\
E-mail: mashford@unimelb.edu.au}
\date{\today}
\maketitle

\begin{abstract}

Pure gravity quantum gravity (pure gravity QG) is non-renormalizable at two and higher loop level. One of the main characteristics of objects occurring in QG is their Lorentz covariance. QG necessitates the analysis of a collection of Lorentz covariant terms of higher and higher order, which is not feasible using renormailization. This paper describes how the regularization technique known as covariant spectral regularization (CSR) can be used to analyze Lorentz covariant 4-tensor valued measures on Minkowski space. The graviton vacuum polarization tensor is cast as such a measure and CSR is used to obtain a well defined and finite function representing this tensor. This is then used to compute the Feynman amplitudes for examples of one-loop, two-loop and $n$-loop graviton-graviton scattering in pure gravity QG. The CSR method results in functions (densities) representing divergent Feynman integrals without requiring renormalization or counterterms.  Therefore it can be used to analyze arbitrarily complicated multiloop QG Feynman diagrams without being crippled by a proliferation of counterterms, as is classical QG.

\end{abstract}

\tableofcontents

\section{Introduction}

't Hooft and Veltman \cite{tHooft} proved that divergences in pure gravity quantum gravity (pure gravity QG) at one loop level can be dealt with using a finite number of counterterms in renormalization, i.e. the theory is renormalizable. 

On the other hand they proved that, in the presence of scalar particles, QG is non-renormalizable at one loop level, and that pure gravity QG without matter particles is non-renormalizable at two and higher loop level when using the classical renormalization technique. 

Deser and van Nieuwenhuizen \cite{Deser1a} proved the non-renormalizability of the Dirac-Einstein system and Deser {\em et al.} \cite{Deser1b} proved the non-renormalizability of Einstein-Yang-Mills interactions at one-loop level, once again, in the context of the classical  approach of renormalization for circumventing divergences.

One attempt to address this problem is to modify the Einstein-Hilbert Lagrangian to a higher derivative Lagrangian. In this approach higher derivative terms involving the curvature are added to the Einstein-Hilbert action. This eliminates the problem of non-renormalizability. However problems with unitarity, causality and in dealing with unphysical ghost states arise \cite{Buonifante}.

Another candidate for a solution to the problem of QG non-renormalizability was supergravity but it has been known for more than 25 years that $D=11$ supergravity is non-renormalizable at 2-loop level \cite{Deser}.

String theory and Loop quantum gravity are both interesting and promising approaches to QG on which a large amount of research has been undertaken.

Some authors believe that GR will act as a universal cutoff for all divergences (including those of GR) while other authors hold the opposite view, that GR is irretrievably un-renormalizable owing to the dimensionality of the gravitational coupling constant \cite{Deser1}.

The approach known as causal perturbation theory is a candidate for a method of carrying out calculations in QG without encountering divergence. Grillo \cite{Grillo} carries out one-loop calculations in QG using causal perturbation theory. While not having divergences in some cases this method has limited applicability and some other problems which we will mention below. 

We have developed a method which we originally called spectral regularization for carrying out computations in QFT without encountering divergence or needing to carry out renormalization. We now call this method covariant spectral regularization (CSR) in order to distinguish it from other techniques called ``spectral regularization".

In the next section we will discuss approaches previously developed for carrying out non-tree-level computations in QFT and in Section~\ref{section:CSR} we will discuss covariant spectral regularization. Then, after an excursion in Section \ref{section:conformally_flat_spacetime} where we discuss the geometry of spacetime, the rest of the paper will show how CSR can be used to compute the graviton vacuum polarization tensor without renormalization or needing counterterms, and how this can be used to compute the Feynman amplitude for multiloop graviton-graviton scattering.  

\section{Methods previously developed for non-tree-level QFT computations}

Tree level computations in QFT, that is computations based on Feynman diagrams with no loops, give quite good results. An example is the M{\o}ller scattering formula for electron-electron scattering. However more accurate predictions can be made by taking into account higher order Feynman diagrams which have loops, leading to what are known as radiative corrections. 

It was found, in the early days of QFT, that, when following the Feynman rules, the integrals involved in radiative corrections are divergent. A number of methods for circumventing the problem caused by these infinities have been developed. We will discuss briefly these methods in the remainder of this section.

The method of dimensional regularization was developed in 1972. In this method divergent parts of the offending integrals are isolated by introducing a fictitious parameter, a deviation $\epsilon$ of the spacetime dimension from 4. The current standard QFT regularization method is $\overline{\text{MS}}$ dimensional regularization. In this method a term involving the parameter $\epsilon$ is subtracted from a divergent integral. This term becomes infinite as $\epsilon\rightarrow0$. Such infinities are cancelled by adding ``counterterms" to the Lagrangian. The counterterms are infinite. 

This process of manipulating and cancelling infinite quantities is called renormalization. The $\overline{\text{MS}}$ regularization/renormalization method has been used extensively in QFT computations in conjunction with the Lagrangian formulation of the SM resulting in many highly accurate predictions which have been thoroughly experimentally verified over the past five decades. 

The prncipal other forms of regularization in QFT are Pauli-Villars, hard cutoff, heat kernel, Gaussian and $\zeta$-function \cite{Schwartz}. They all require renormalization.

The method of causal perturbation theory (causal pQFT), and, in particular, causal perturbative QED (causal pQED) was developed in which the process of renormalization is not carried out. However a less problematic process of ``normalization" needs to be carried out. 
Causal pQED had its origin in work of Bogoliubov and others but its role in formulating a divergence-free representation of QFT became apparent through the work of Epstein and Glaser \cite{Epstein} in 1973. It was further developed by a number of researchers and described in the book by Scharf (Scharf, 1995 \cite{Scharf}).

Causal pQFT has a number of problems or difficulties which can be summarized as follows \cite{Scharf}.

\begin{itemize}
\item In full QED the fundamental causality axiom concerning distribution splitting for causal pQED does not have a direct physical foundation.
\item Incorrect distribution splitting which may occur if the distributions don't have the right singular order, leads to ultraviolet divergences. 
\item For every graph with singular order $\omega\geq0$, a polynomial of degree $\omega$ remains undetermined in momentum space. Due to this situation a process called {\em normalization} needs to be carried out. 
\item As in renormalization, normalization requires free parameters. There are cases where the number of required free parameters increases with the order of the perturbation expansion without bound in which case the theory is non-normalizable. 
\item Infrared divergences appear, at a much later stage, namely, when one calculates observable quantities in the adiabatic limit.
\end{itemize}

\section{Covariant spectral regularization\label{section:CSR}}

Relatively recently the method of covariant spectral regularization (CSR) has been developed and applied to a number of practical problems in pQFT and, in particular, pQED. It operates by interpreting the objects defined by Feynman integrals as Lorentz covariant Borel complex tensor valued measures, computing their spectra using a spectral calculus and hence computing densities associated with the objects which can be used in computations involving the Feynman integrals. No renormalization or normalization is required and the computations are UV and IR divergence-free. 

CSR is introduced in Ref. \cite{Springer}. The first step in interpreting a Feynman integral using CSR is to assume ``pretend" that the function purportedly defined by the integral expression exists and to treat it as a density for a measure on Minkowski space. 

Then one carries out a formal argument where the only non-rigorous step is the interchange of the order of integration at one point. Note that if one is to proceed by a series of steps from something that is not well defined (e.g. a Feynman integral) to something that is well defined (a Borel complex tensor valued measure) then not all steps can be rigorous, there must be one or more ``formal" steps. In this case the formal step is the interchange of an order of integration. 

A rigorous step in this sequence of steps is to interpret propagator like functions as measures on Minkowski space. Specifically one makes the identification
\[ \frac{1}{p^2-m^2+i\epsilon}\,dp\leftrightarrow-i\pi\Omega_m^{\pm}(dp), \]
where $m\geq0$ is the particle mass and $\Omega_m^{\pm}$ is the standard (or canonical) measure on Minkowski space supported on the mass $m$ hyperboloid (or cone) $H_m^{\pm}=\pm\{p\in{\bf R}^4:p^2=m^2,p^0\geq0\}$, given by
\[ \Omega_m^{\pm}(\Gamma)=\int_{{\vct p}\in{\bf R}^3}\chi_{\Gamma}((\pm\omega_m({\vct p}),{\vct p}))\omega_m({\vct p})^{-1}\,d{\vct p}, \]
in which $\chi_{\Gamma}$ is the characteristic function (indicator function) for the set $\Gamma$ given by
\[ \chi_{\Gamma}(p)=\left\{\begin{array}{l}
1\text{ if }p\in\Gamma\\
0\text{ otherwise,}
\end{array}\right. \]
and
\[ \omega_m({\vct p})=(m^2+{\vct p}^2)^{\frac{1}{2}}. \]
The justification and derivation for this identification is given in Ref. \cite{Symmetry}. 

After this formal argument one arrives at a valid well defined Borel complex tensor valued measure on Minkowski space. 

It is then necessary to prove that the measure so obtained is Lorentz covariant.

Then one uses a rigorously well defined computation using a spectral calculus \cite{Symmetry} to compute the spectrum of the measure from which one can compute the density for the measure, and hence a well defined analytic function representing the Feynman integral. 

The technique has been used in practical applications such as the computation of the Uehling contribution to the Lamb shift for the hydrogen atom \cite{IJMPA,NPB}, the computation of the anomalous magnetic moment of the electron \cite{MSQUARE}, the computation of the next-to-leading-order (NLO) correction to the high energy limit of the cross section for muon production from electron-positron annihilation  \cite{MSQUARE_2025} and the computation of running couplings for QED and QCD which do not manifest Landau poles \cite{NPB,IJMPA_2025}. All the classical tests of radiative corrections in QED are satisfied when using CSR \cite{IJMPA_2025}. 

The results of using CSR agree precisely (on the examples for which it has been tested) with those of using $\overline{\text{MS}}$ except that CSR does not have the, what are usually described as ``artefacts", such as UV divergences, requisite counterterms, Landau poles, IR divergences which need to be cancelled by soft photons etc., that are present when using $\overline{\text{MS}}$.

In this paper the method will be further developed in the context of Lorentz covariant 4-tensor valued measures and  illustrated in our QG computations below.

\section{Locally conformally flat spacetimes\label{section:conformally_flat_spacetime}}

If $X$ is a ($C^{\infty}$) Lorentzian manifold then $X$ is locally conformally flat if and only if for each point $x\in X$ there exists a chart $\phi:U\rightarrow V$ for $X$ where $U\subset X$ is an open neighborhood of $x$ and $V$ is an open subset of ${\bf R}^4$, such that the metric tensor $g_{\mu\nu}$ for $X$ has, in the chart $\phi$, the form
\[ g_{\mu\nu}(\xi)=\lambda(\xi)\eta_{\mu\nu},\forall\xi\in V, \]
for some $\lambda\in C^{\infty}(V,(0,\infty))$,where $\eta=\text{diag}(1,-1,-1,-1)$.

There is a natural one-to-one correspondence between the set of equivalence classes of conformally equivalent locally conformally flat spacetimes and the set of what are known as M\"{o}bius structures \cite{AMP}. M\"{o}bius structures are defined by atlases of conformally compatible charts and do not have explicit metrics.

A canonical structure associated with any M\"{o}bius structure is a principal bundle with structure group the orthochronous Lorentz group. Thus, without explicitly imposing a metric, the tangent space has the structure of Minkowski space and its Lorentz symmetry.

A number of spacetimes of interest in general relativity (GR) are conformally flat. Firstly any globally conformally flat GR spacetime is locally conformally flat. Four notable examples of globally conformally flat GR spacetimes are Minkowski spacetime, the de Sitter spacetime which describes a universe with a positive cosmological constant, the anti-de Sitter (ADS) space which is studied extensively in string theory and the Friedman-Lema\^{i}tre-Robertson-Walker (FLRW) spacetime. 

The FLRW metric describes an expanding or contracting universe and is commonly used to model the universe on the large scale. 

Another example is the interior Schwarzschild metric which describes a constant density sphere and is commonly used to model stars and planets. 

However not all GR spacetimes are locally conformally flat. The vacuum Schwarzs-child solution is a notable example. 

In QFT there is no such thing as completely empty space. The QFT vacuum contains a sea of virtual particles. Therefore it might be that the vacuum Schwarzschild solution is a limiting case, an effective theory construct, associated with a non-empty quantum vacuum state which might be modeled by a locally conformally flat spacetime region such as that defined by, for example, the interior Schwarzschild metric.

The importance of restricting attention to locally conformally flat spacetimes is, apart from the naturalness of modeling spacetime as as a M\"{o}bius structure, i.e. as a locally Minkowskian causal structure \cite{AMP}, that M\"{o}bius structures have (up to a canonical removable dilatation) Minkowski space as tangent space. i.e. the tangent space possesses the metric structure of Minkowski space without having to impose it. 

By contrast, in GR spacetime is modeled as a smooth 4-manifold which has, as tangent space, ${\bf R}^4$ with its vector space structure but no metric structure. Then a metric structure is imposed in the form of a metric tensor field resulting in a Lorentzian manifold.

A common approach to attempting to quantize gravity is to, in some sense, quantize the metric tensor $g_{\mu\nu}$, that is to treat it (or its deviation $h_{\mu\nu}=g_{\mu\nu}-\eta_{\mu\nu}$ from the flat spacetime metric tensor) as a quantum field on a par with the electromagnetic field in QED e.g. see Ref.~\cite{Mandelstam}. The problem with this is that, in GR, the metric tensor defines the geometry of spacetime, i.e. the {\em arena} in which the fields operate so it is somewhat contradictory, or circular, to treat the metric as a dynamic field. 

In our work there is no distinguished metric tensor. The paramount structures are principal bundles, one with structure group $U(2,2)$ and related bundles and another with structure group $(0,\infty)\times O(1,3)^{\uparrow}$ (where $O(1,3)^{\uparrow}$ denotes the orthochronous Lorentz group) \cite{AMP,IJMPA}. There is no canonical metric but there are canonical symmetry group structures.

We consider, in common with GR, that spacetime has a continuum nature at all scales. Then the notion of tangent space to the spacetime manifold is well defined and non-problematic. It is, within the tangent space at any point $x\in X$ in spacetime, that QFT momentum variables range, providing a flat arena for the Feynman amplitude functions and the consequent structure of the SM.

By considering locally conformally flat spacetimes (M\"{o}bius structures) Lorentz covariance (on the tangent space to spacetime $X$ at any point $x\in X$) becomes of significance in a natural way without having to add further structure to impose it.

\section{Lorentz covariant 4-tensor valued measures on Minkowski space}

\subsection{Lorentz covariant measures and their densities}

Let ${\mathcal B}({\bf R}^4)$ denote the Borel $\sigma$-algebra for Minkowski space (${\bf R}^4$) with respect to the Euclidean topology for ${\bf R}^4$ and ${\mathcal B}_0({\bf R}^4)=\{\Gamma\in{\mathcal B}({\bf R}^4):\text{ $\Gamma$ is relatively compact}\}$. By a complex $n$-tensor valued measure we will mean, in this paper, an array of maps $\Pi^{\mu_1\ldots\mu_n}:{\mathcal B}_0({\bf R}^4)\rightarrow{\bf C}$ where $\mu_1,\ldots,\mu_n\in\{0,1,2,3\}$ such that $\Pi(\emptyset)=0$ and $\Pi|_K$ is $\sigma$-additive for all compact $K\subset{\bf R}^4$ where
\[ \Pi|_K(\Gamma)=\Pi(\Gamma\cap K),\forall\Gamma\in{\mathcal B}_0({\bf R}^4). \]
A contravariant complex $n$-tensor valued measure $\Pi^{\mu_1\ldots\mu_n}:{\mathcal B}_0({\bf R}^4)\rightarrow{\bf C}$ will be said to be Lorentz covariant if
\[ \Pi^{\mu_1\ldots\mu_n}(\Lambda\Gamma)=\Lambda^{\mu_1}{}_{\rho_1}\ldots\Lambda^{\mu_n}{}_{\rho_n}\Pi^{\rho_1\ldots\rho_n}(\Gamma),\forall\Lambda\in O(1,3)^{\uparrow},\Gamma\in{\mathcal B}_0({\bf R}^4). \]
(Here $\Lambda\Gamma$ denotes the set $\Lambda\Gamma=\{\Lambda q:q\in\Gamma\}$.)
In particular, a 4-tensor valued measure $\Pi^{\mu\nu\rho\sigma}:{\mathcal B}_0({\bf R}^4)\rightarrow{\bf C}$ is Lorentz covariant if
\[ \Pi^{\mu\nu\rho\sigma}(\Lambda\Gamma)=\Lambda^{\mu}{}_{\alpha}\Lambda^{\nu}{}_{\beta}\Lambda^{\rho}{}_{\gamma}\Lambda^{\sigma}{}_{\delta}\Pi^{\alpha\beta\gamma\delta}(\Gamma),\forall\Lambda\in O(1,3)^{\uparrow},\Gamma\in{\mathcal B}_0({\bf R}^4). \]
A contravariant $n$-tensor valued function $\Pi^{\mu_1\ldots\mu_n}$ defined on a domain $D\in{\mathcal B}({\bf  R}^4)$ (such that $O(1,3)^{\uparrow}D=D$) will be said to be Lorentz covariant if
\[ \Pi^{\mu_1\ldots\mu_n}(\Lambda q)=\Lambda^{\mu_1}{}_{\rho_1}\ldots\Lambda^{\mu_n}{}_{\rho_n}\Pi^{\rho_1\ldots\rho_n}(q),\forall\Lambda\in O(1,3)^{\uparrow},q\in D. \]

\newtheorem{lemma}{Lemma}
\begin{lemma}
Let $\Pi^{\mu\nu\rho\sigma}$ be a Lorentz covariant 4-tensor valued measure which is associated with (i.e. generated by) a density which we also denote as $\Pi^{\mu\nu\rho\sigma}$. Then, for all $\Lambda\in O(1,3)^{\uparrow}$ the density $\Pi^{\mu\nu\rho\sigma}$ satisfies
\begin{equation}\label{eq:density_equation}
\Pi^{\mu\nu\rho\sigma}(\Lambda q)=\Lambda^{\mu}{}_{\alpha}\Lambda^{\nu}{}_{\beta}\Lambda^{\rho}{}_{\gamma}\Lambda^{\sigma}{}_{\delta}\Pi^{\alpha\beta\gamma\delta}(q),
\end{equation}
for almost all $q\in{\bf R}^4$.
\end{lemma}
{\noindent{\bf Proof}}
We have
\[ \Pi^{\mu\nu\rho\sigma}(\Gamma)=\int_{\Gamma}\Pi^{\mu\nu\rho\sigma}(q)\,dq,\forall\Gamma\in{\mathcal B}_0({\bf R}^4). \]
Thus, by the Lorentz invariance of Lebesgue measure on ${\bf R}^4$, we have that for any $\Lambda\in O(1,3)^{\uparrow}$
\[ \Pi^{\mu\nu\rho\sigma}(\Lambda\Gamma)=\int_{\Lambda\Gamma}\Pi^{\mu\nu\rho\sigma}(q)\,dq=\int_{\Gamma}\Pi^{\mu\nu\rho\sigma}(\Lambda q)\,dq. \]
Also, by hypothesis,
\begin{align*}
\Pi^{\mu\nu\rho\sigma}(\Lambda\Gamma)=&\Lambda^{\mu}{}_{\alpha}\Lambda^{\nu}{}_{\beta}\Lambda^{\rho}{}_{\gamma}\Lambda^{\sigma}{}_{\delta}\Pi^{\alpha\beta\gamma\delta}(\Gamma)\\
=&\Lambda^{\mu}{}_{\alpha}\Lambda^{\nu}{}_{\beta}\Lambda^{\rho}{}_{\gamma}\Lambda^{\sigma}{}_{\delta} \int_{\Gamma}\Pi^{\alpha\beta\gamma\delta}(q)\,dq,
\end{align*}
for all $\Gamma\in{\mathcal B}_0({\bf R}^4)$. Thus for each $\Lambda\in O(1,3)^{\uparrow}$
\[ \int_{\Gamma}(\Pi^{\mu\nu\rho\sigma}(\Lambda q)-\Lambda^{\mu}{}_{\alpha}\Lambda^{\nu}{}_{\beta}\Lambda^{\rho}{}_{\gamma}\Lambda^{\sigma}{}_{\delta}\Pi^{\alpha\beta\gamma\delta}(q))\,dq=0,\forall\Gamma\in{\mathcal B}_0({\bf R}^4). \] 
Therefore for each $\Lambda\in O(1,3)^{\uparrow}$
\[ \Pi^{\mu\nu\rho\sigma}(\Lambda q)=\Lambda^{\mu}{}_{\alpha}\Lambda^{\nu}{}_{\beta}\Lambda^{\rho}{}_{\gamma}\Lambda^{\sigma}{}_{\delta}\Pi^{\alpha\beta\gamma\delta}(q), \]
for almost all $q\in{\bf R}^4$. $\Box$

If the density $\Pi$ is continuous then Eq.~\ref{eq:density_equation} holds for all $q\in{\bf R}^4$.

Conversely it is straightforward to show that if $\Pi^{\mu\nu\rho\sigma}:{\bf R}^4\rightarrow{\bf C}$ is a (locally integrable) Lorentz covariant complex tensor valued function then the measure $\Pi$ defined by
\[ \Pi^{\mu\nu\rho\sigma}(\Gamma)=\int_{\Gamma}\Pi^{\mu\nu\rho\sigma}(q)\,dq, \]
is a Lorentz covariant Borel complex tensor valued measure on Minkowski space. 

(No confusion need arise through using the same symbol for a function and the associated measure, the intended meaning in any case can be determined by the argument of the symbol.)

\subsection{Causal measures and their spectra}

Suppose that $\Pi$ is a Lorentz covariant complex 4-tensor valued function which is causal, i.e. is supported on $\{q\in{\bf R}^4:q^2\geq0,q^0\geq0\}$. Then $\Pi$ is determined (up to its values on the cone $C=\{p\in{\bf R}^4:p^2=0,p^0\geq0\}$ which is a set of Lebesgue measure zero) by its values on $\{(m,{\vct0}):m>0\}$.
Define $\lambda^{\mu\nu\rho\sigma}:(0,\infty)\rightarrow{\bf C}$ by
\[ \lambda^{\mu\nu\rho\sigma}(m)=\Pi^{\mu\nu\rho\sigma}((m,{\vct0}))\text{ for }m>0. \]
Thus the most general Lorentz covariant 4-tensor valued measure is determined by a function which takes, for each $m>0$, values $\lambda^{\mu\nu\rho\sigma}(m)$ which lie in the 256 dimensional space ${\bf C}^4\times{\bf C}^4\times{\bf C}^4\times{\bf C}^4\cong{\bf C}^{256}$. We call the $\lambda^{\mu\nu\rho\sigma}$ functions spectral functions for the measure $\Pi^{\mu\nu\rho\sigma}$. 

However not all spectral functions $m\mapsto\lambda(m)$ in this space correspond to different measures because we have the following 
\begin{align}\label{eq:spectral_property}
\lambda^{\mu\nu\rho\sigma}(m)=&\Pi^{\mu\nu\rho\sigma}((m,{\vct0}))\nonumber\\
=&\Pi^{\mu\nu\rho\sigma}(R(m,{\vct 0}))\nonumber\\
=&R^{\mu}{}_{\alpha}R^{\nu}{}_{\beta}R^{\rho}{}_{\gamma}R^{\sigma}{}_{\delta}\Pi^{\alpha\beta\gamma\delta}((m,{\vct0}))\nonumber\\
=&R^{\mu}{}_{\alpha}R^{\nu}{}_{\beta}R^{\rho}{}_{\gamma}R^{\sigma}{}_{\delta}\lambda^{\alpha\beta\gamma\delta}(m),
\end{align}
for all $m>0,R\in\text{Rotations}$ where $\text{Rotations}\subset O(1,3)^{\uparrow}$ is the rotation subgroup of $O(1,3)^{\uparrow}$. 
The most general solution to Eq.~\ref{eq:spectral_property} is a linear combination of 10 independent functions of $m$. To show this observe that any rotation $R\in$ Rotations acts as the identity on all time components of a tensor $\lambda^{\mu\nu\rho\sigma}$ and, to satisfy Eq.~\ref{eq:spectral_property} for all rotations $R$ each spatial sector must be an isotropic tensor under $O(3)$. We consider the following cases.
\begin{itemize}
\item No indices are 0. Then there are three independent isotropic tensors, $\delta_{ij}\delta_{kl},\delta_{ik}\delta_{jl}$ and $\delta_{il}\delta_{jk}$.
\item One index is 0. Then there are no isotropic tensors. (Note that the Levi-Civita symbol $\epsilon_{jkl}$ is isotropic with respect to $SO(3)$ but not $O(3)$.)
\item Two indices are 0. Then there is one isotropic tensor $\delta_{kl}$ and $\left(\begin{array}{c}
4\\
2
\end{array}\right)=6$ ways to place the two 0 indices.
\item Three indices are 0. There are no isotropic rank 1 tensors (vectors) for $O(3)$.
\item All indices are 0. Then this is a scalar contributing one dimension. 
\end{itemize}

Therefore the set of spectral functions is made up of functions from $(0,\infty)$ to a complex vector space of 10 dimensions that is, a 10 dimensional complex hyperplane in ${\bf C}^{256}$.

\section{Canonical form for Lorentz covariant 4-tensor valued measures on Minkowski space}

Let $\Omega_m=\Omega_m^{+}$ for $m\geq0$. Define the canonical form to be given by
\begin{align}\label{eq:canonical_form}
\Pi^{\mu\nu\rho\sigma}(\Gamma)=&\int_{m=0}^{\infty}\int_{q\in\Gamma}(q^2\eta^{\mu\nu}\eta^{\rho\sigma}\sigma_1(m)+q^2\eta^{\mu\rho}\eta^{\nu\sigma}\sigma_2(m)+q^2\eta^{\mu\sigma}\eta^{\nu\rho}\sigma_3(m)+\nonumber\\
&\eta^{\mu\nu}q^{\rho}q^{\sigma}\sigma_4(m)+\eta^{\mu\rho}q^{\nu}q^{\sigma}\sigma_5(m)+\eta^{\mu\sigma}q^{\nu}q^{\rho}\sigma_6(m)+\eta^{\nu\rho}q^{\mu}q^{\sigma}\sigma_7(m)+\nonumber\\
&\eta^{\nu\sigma}q^{\mu}q^{\rho}\sigma_8(m)+\eta^{\rho\sigma}q^{\mu}q^{\nu}\sigma_9(m)+q^{\mu}q^{\nu}q^{\rho}q^{\sigma}\sigma_{10}(m))\,\Omega_m(dq)\,dm\\
=&\sum_{i=1}^{10}\Pi^{\mu\nu\rho\sigma}_i(\Gamma),\nonumber
\end{align}
say, where $\sigma_i:(0,\infty)\rightarrow{\bf C}$ are locally integrable spectral functions.  

It is straightforward to show that, for all $i=1,\ldots,10$, $\Pi_i^{\mu\nu\rho\sigma}(\emptyset)=0$.
\begin{lemma} Suppose that $f:{\bf R}^4\rightarrow[0,\infty)$ is locally integrable. and $\Gamma\in{\mathcal B}_0({\bf R}^4)$ is compact.  Then
\[ \int_{m=0}^{\infty}\int_{q\in\Gamma}f(q)\,\Omega_m(dq)\,dm<\infty. \] 
\end{lemma}
{\noindent\bf Proof}
There exists $R>0$ such that $|q^0|<R$ for all $q=(q^0,{\vct q})\in\Gamma$. Therefore  
\[ \int_{q\in\Gamma}f(q)\,\Omega_m(dq)=\int_{{\vct q}\in{\bf R}^3}\chi_{\Gamma}((\omega_m({\vct q}),{\vct q}))f((\omega_m({\vct q}),{\vct q}))\omega_m({\vct q})^{-1}\,d{\vct q}=0,\forall m>R. \]
Thus, letting $K=\{{\vct q}\in{\bf R}^3:\exists q^0\in{\bf R},(q^0,{\vct q})\in\Gamma\}$, we have that
\[ \int_{m=0}^{\infty}\int_{q\in\Gamma}f(q)\,\Omega_m(dq)\,dm=\int_{m=0}^R\int_{{\vct q}\in K}\chi_{\Gamma}((\omega_m({\vct q}),{\vct q}))f(q)\omega_m({\vct q})^{-1}\,d{\vct q}\,dm<\infty, \] 
because the integrand is locally integrable and the domain of integration is compact. $\Box$

This proves, in particular, that all the $\Pi_i^{\mu\nu\rho\sigma}:{\mathcal B}_0({\bf R}^4)\rightarrow{\bf C}$. It is straightforward to show that $\Pi_i^{\mu\nu\rho\sigma}|_K$ is $\sigma$-additive for any compact $K\subset{\bf R}^4$. Therefore $\Pi_i^{\mu\nu\rho\sigma}$ is a complex 4-tensor valued measure for all $i=1,\ldots,10$. Thus $\Pi^{\mu\nu\rho\sigma}$ is a complex 4-tensor valued measure.
\begin{lemma} $\Pi_1^{\mu\nu\rho\sigma}$ is Lorentz covariant.
\end{lemma}
{\noindent{\bf Proof}}
\[ \Pi_1^{\mu\nu\rho\sigma}(\Gamma)=\int_{m=0}^{\infty}\int_{q\in\Gamma}q^2\eta^{\mu\nu}\eta^{\rho\sigma}\sigma_1(m)\,\Omega_m(dq)\,dm. \]
Therefore, by the Lorentz invariance of the measures $\Omega_m$,
\begin{align*}
\Pi_1^{\mu\nu\rho\sigma}(\Lambda\Gamma)=&\int_{m=0}^{\infty}\int_{q\in\Gamma}(\Lambda q)^2\eta^{\mu\nu}\eta^{\rho\sigma}\sigma_1(m)\,\Omega_m(dq)\,dm\\
=&\int_{m=0}^{\infty}\int_{q\in\Gamma}q^2\eta^{\mu\nu}\eta^{\rho\sigma}\sigma_1(m)\,\Omega_m(dq)\,dm=\Pi_1^{\mu\nu\rho\sigma}(\Gamma).
\end{align*}
But also, by the fundamental defining property of the Lorentz transformations,
\begin{align*}
\Lambda^{\mu}{}_{\alpha}\Lambda^{\nu}{}_{\beta}\Lambda^{\rho}{}_{\gamma}\Lambda^{\sigma}{}_{\delta}\Pi_1^{\alpha\beta\gamma\delta}(\Gamma)=&\eta^{\alpha\beta}\eta^{\gamma\delta}\Lambda^{\mu}{}_{\alpha}\Lambda^{\nu}{}_{\beta}\Lambda^{\rho}{}_{\gamma}\Lambda^{\sigma}{}_{\delta}\int_{m=0}^{\infty}\int_{q\in\Gamma}q^2\sigma_1(m)\,\Omega_m(dq)\,dm\\
=&\eta^{\mu\nu}\eta^{\rho\sigma}\int_{m=0}^{\infty}\int_{q\in\Gamma}q^2\sigma_1(m)\,\Omega_m(dq)\,dm=\Pi_1^{\mu\nu\rho\sigma}(\Gamma).
\end{align*}
$\Box$

It is straightforward to show, using similar calculations, that $\Pi_i$ all define Lorentz covariant Borel complex 4-tensor valued measures on Minkowski space.  

Now a vector $q\in\{q\in{\bf R}^4:q^2>0,q^0>0\}$ can be considered to be a function of $m$ (the parameter specifying in which orbit of $O(1,3)^{\uparrow}$ $q$ lies) and ${\vct q}$ (the projection of $q$ onto ${\bf R}^3$) according to
\[ q=q(m,{\vct q})=(\omega_m({\vct q}),{\vct q}), m>0,{\vct q}\in{\bf R}^3. \]
The Jacobian of the transformation is
\[ J(m,{\vct q})=m\omega_m({\vct q})^{-1}. \]
Also
\[ m=(q(m,{\vct q}))^{\frac{1}{2}}=(q^2)^{\frac{1}{2}}. \]
Therefore, by Eq.~\ref{eq:canonical_form},
\begin{align*}
\Pi^{\mu\nu\rho\sigma}(\Gamma)=&\int_{m=0}^{\infty}\int_{{\vct q}\in{\bf R}^3}\chi_{\Gamma}((\omega_m({\vct q}),{\vct q}))(m^2\eta^{\mu\nu}\eta^{\rho\sigma}\sigma_1(m)+m^2\eta^{\mu\rho}\eta^{\nu\sigma}\sigma_2(m)+\\
&m^2\eta^{\mu\sigma}\eta^{\nu\rho}\sigma_3(m)+\eta^{\mu\nu}q^{\rho}q^{\sigma}\sigma_4(m)+\eta^{\mu\rho}q^{\nu}q^{\sigma}\sigma_5(m)+\eta^{\mu\sigma}q^{\nu}q^{\rho}\sigma_6(m)+\\
&\eta^{\nu\rho}q^{\mu}q^{\sigma}\sigma_7(m)+\eta^{\nu\sigma}q^{\mu}q^{\rho}\sigma_8(m)+\eta^{\rho\sigma}q^{\mu}q^{\nu}\sigma_9(m)+\\
&q^{\mu}q^{\nu}q^{\rho}q^{\sigma}\sigma_{10}(m))\,\omega_m({\vct q})^{-1}\,d{\vct q}\,dm\\
=&\int_{q^2>0,q^0>0}\chi_{\Gamma}(q)(m^2\eta^{\mu\nu}\eta^{\rho\sigma}\sigma_1(m)+m^2\eta^{\mu\rho}\eta^{\nu\sigma}\sigma_2(m)+m^2\eta^{\mu\sigma}\eta^{\nu\rho}\sigma_3(m)+\\
&\eta^{\mu\nu}q^{\rho}q^{\sigma}\sigma_4(m)+\eta^{\mu\rho}q^{\nu}q^{\sigma}\sigma_5(m)+\eta^{\mu\sigma}q^{\nu}q^{\rho}\sigma_6(m)+\eta^{\nu\rho}q^{\mu}q^{\sigma}\sigma_7(m)+\\
&\eta^{\nu\sigma}q^{\mu}q^{\rho}\sigma_8(m)+\eta^{\rho\sigma}q^{\mu}q^{\nu}\sigma_9(m)+q^{\mu}q^{\nu}q^{\rho}q^{\sigma}\sigma_{10}(m))\\
&\omega_m({\vct q})^{-1}(m\omega_m({\vct q})^{-1})^{-1}\,dq\\
=&\int_{q^2>0,q^0>0}\chi_{\Gamma}(q)(m\eta^{\mu\nu}\eta^{\rho\sigma}\sigma_1(m)+m\eta^{\mu\rho}\eta^{\nu\sigma}\sigma_2(m)+m\eta^{\mu\sigma}\eta^{\nu\rho}\sigma_3(m)+\\
&m^{-1}\eta^{\mu\nu}q^{\rho}q^{\sigma}\sigma_4(m)+m^{-1}\eta^{\mu\rho}q^{\nu}q^{\sigma}\sigma_5(m)+m^{-1}\eta^{\mu\sigma}q^{\nu}q^{\rho}\sigma_6(m)+\nonumber\\
&m^{-1}\eta^{\nu\rho}q^{\mu}q^{\sigma}\sigma_7(m)+m^{-1}\eta^{\nu\sigma}q^{\mu}q^{\rho}\sigma_8(m)+m^{-1}\eta^{\rho\sigma}q^{\mu}q^{\nu}\sigma_9(m)+\\
&m^{-1}q^{\mu}q^{\nu}q^{\rho}q^{\sigma}\sigma_{10}(m))\,dq,
\end{align*}
where $m=m({\vct q})=(q^2)^{\frac{1}{2}}$. Therefore the density for the measure $\Gamma\mapsto\Pi^{\mu\nu\rho\sigma}(\Gamma)$ given by Eq.~\ref{eq:canonical_form} is
\begin{align}\label{eq:canonical_form_density}
\Pi^{\mu\nu\rho\sigma}(q)=&Q(\eta^{\mu\nu}\eta^{\rho\sigma}\sigma_1(Q)+\eta^{\mu\rho}\eta^{\nu\sigma}\sigma_2(Q)+\eta^{\mu\sigma}\eta^{\nu\rho}\sigma_3(Q))+\nonumber\\
&Q^{-1}(\eta^{\mu\nu}q^{\rho}q^{\sigma}\sigma_4(Q)+\eta^{\mu\rho}q^{\nu}q^{\sigma}\sigma_5(Q)+\eta^{\mu\sigma}q^{\nu}q^{\rho}\sigma_6(Q)+\nonumber\\
&\eta^{\nu\rho}q^{\mu}q^{\sigma}\sigma_7(Q)+\eta^{\nu\sigma}q^{\mu}q^{\rho}\sigma_8(Q)+\eta^{\rho\sigma}q^{\mu}q^{\nu}\sigma_9(Q)+\nonumber\\
&q^{\mu}q^{\nu}q^{\rho}q^{\sigma}\sigma_{10}(Q)),
\end{align}
for $q^2>0,q^0>0$ and 0 otherwise where $Q=(q^2)^{\frac{1}{2}}$. Thus, since $((m,{\vct0})^2)^{\frac{1}{2}}=m$ the spectral function $\lambda^{\mu\nu\rho\sigma}$ associated with the canonical form is
\begin{align*}
&\lambda^{\mu\nu\rho\sigma}(m)=\Pi^{\mu\nu\rho\sigma}((m,{\vct0}))\\
=&m(\eta^{\mu\nu}\eta^{\rho\sigma}\sigma_1(m)+\eta^{\mu\rho}\eta^{\nu\sigma}\sigma_2(m)+\eta^{\mu\sigma}\eta^{\nu\rho}\sigma_3(m))+\\
&m^{-1}(\eta^{\mu\nu}(m,{\vct0})^{\rho}(m,{\vct0})^{\sigma}\sigma_4(m)+\eta^{\mu\rho}(m,{\vct0})^{\nu}(m,{\vct0})^{\sigma}\sigma_5(m)+\\
&\eta^{\mu\sigma}(m,{\vct0})^{\nu}(m,{\vct0})^{\rho}\sigma_6(m)+\eta^{\nu\rho}(m,{\vct0})^{\mu}(m,{\vct0})^{\sigma}\sigma_7(m)+\\
&\eta^{\nu\sigma}(m,{\vct0})^{\mu}(m,{\vct0})^{\rho}\sigma_8(m)+\eta^{\rho\sigma}(m,{\vct0})^{\mu}(m,{\vct0})^{\nu}\sigma_9(m)+\\
&(m,{\vct0})^{\mu}(m,{\vct0})^{\nu}(m,{\vct0})^{\rho}(m,{\vct0})^{\sigma}\sigma_{10}(m))\\
=&m\left(\eta^{\mu\nu}\eta^{\rho\sigma}\sigma_1(m)+\eta^{\mu\rho}\eta^{\nu\sigma}\sigma_2(m)+\eta^{\mu\sigma}\eta^{\nu\rho}\sigma_3(m)+\eta^{\mu\nu}\eta^{0\rho}\eta^{0\sigma}\sigma_4(m)+\right.\\
&\eta^{\mu\rho}\eta^{0\nu}\eta^{0\sigma}\sigma_5(m)+\eta^{\mu\sigma}\eta^{0\nu}\eta^{0\rho}\sigma_6(m)+\eta^{\nu\rho}\eta^{0\mu}\eta^{0\sigma}\sigma_7(m)+\\
&\left.\eta^{\nu\sigma}\eta^{0\mu}\eta^{0\rho}\sigma_8(m)+\eta^{\rho\sigma}\eta^{0\mu}\eta^{0\nu}\sigma_9(m)+m^2\eta^{0\mu}\eta^{0\nu}\eta^{0\rho}\eta^{0\sigma}\sigma_{10}(m)\right).
\end{align*}
Therefore the set of canonical measure spectra is associated with the set of continuous  maps from $(0,\infty)$ to a 10 dimensional hyperplane in ${\bf C}^{256}$. Thus the canonical measure is the most general form of a causal Lorentz covariant Borel 4-tensor valued measure on Minkowski space which can be induced by a locally integrable density.

\section{Spectral calculus for Lorentz covariant 4-tensor valued measures on Minkowski space\label{section:spectral_calculus}}

\subsection{General case}

Consider an arbitrary causal Lorentz covariant Borel 4-tensor valued measure $\Pi^{\mu\nu\rho\sigma}$ on Minkowski space. Then, from above, if it can be induced by a locally integrable density then it must be in the canonical form defined by Eq.~\ref{eq:canonical_form}. We would like to compute the spectral functions $\sigma_i$ so that we can recover the density $\Pi^{\mu\nu\rho\sigma}$ using Eq.~\ref{eq:canonical_form_density}. To do so we will use a spectral calculus for Lorentz covariant complex 4-tensor valued measures analogous to the spectral calculus defined previously \cite{Symmetry,IJMPA,NPB} for Lorentz covariant complex 2-tensor valued measures on Minkowski space.

Let $\Gamma(a,b,\epsilon)\subset{\bf R}^4$ for $0<a<b,\epsilon>0$ be the open hyperbolic cylinder defined by
\[ \Gamma(a,b,\epsilon)=\bigcup_{m\in(a,b)}S(m,\epsilon), \]
where 
\[ S(m,\epsilon)=\{q\in{\bf R}^4:q^2=m^2,q^0>0,|{\vct q}|<\epsilon\}. \]
Define the function $g^{\mu\nu\rho\sigma}:\{(a,b,\epsilon):0<a<b,\epsilon>0\}\rightarrow{\bf C}$ for $\mu,\nu,\rho,\sigma=0,1,2,3$ by
\[ g^{\mu\nu\rho\sigma}(a,b,\epsilon)=\Pi^{\mu\nu\rho\sigma}(\Gamma(a,b,\epsilon)). \]
We assume that we have managed to compute the functions $g^{\mu\nu\rho\sigma}$ (we will do this computation below for the graviton vacuum polarization function). We now want to compute the spectral functions $\sigma_i$ from which we can determine the density for the measure.

Since $\Pi^{\mu\nu\rho\sigma}$ must be in canonical form we can compute
\begin{align*}
g^{\mu\nu\rho\sigma}(a,b,\epsilon)=&\Pi^{\mu\nu\rho\sigma}(\Gamma(a,b,\epsilon))\\
=&\int_{m=0}^{\infty}\int_{q\in\Gamma(a,b,\epsilon)}(q^2\eta^{\mu\nu}\eta^{\rho\sigma}\sigma_1(m)+q^2\eta^{\mu\rho}\eta^{\nu\sigma}\sigma_2(m)+q^2\eta^{\mu\sigma}\eta^{\nu\rho}\sigma_3(m)+\\
&\eta^{\mu\nu}q^{\rho}q^{\sigma}\sigma_4(m)+\eta^{\mu\rho}q^{\nu}q^{\sigma}\sigma_5(m)+\eta^{\mu\sigma}q^{\nu}q^{\rho}\sigma_6(m)+\eta^{\nu\rho}q^{\mu}q^{\sigma}\sigma_7(m)+\\
&\eta^{\nu\sigma}q^{\mu}q^{\rho}\sigma_8(m)+\eta^{\rho\sigma}q^{\mu}q^{\nu}\sigma_9(m)+q^{\mu}q^{\nu}q^{\rho}q^{\sigma}\sigma_{10}(m))\,\Omega_m(dq)\,dm\\
=&\int_{m=0}^{\infty}\int_{q\in{\bf R}^4}\chi_{\Gamma(a,b,\epsilon)}(q)(q^2\eta^{\mu\nu}\eta^{\rho\sigma}\sigma_1(m)+q^2\eta^{\mu\rho}\eta^{\nu\sigma}\sigma_2(m)+q^2\eta^{\mu\sigma}\eta^{\nu\rho}\sigma_3(m)+\\
&\eta^{\mu\nu}q^{\rho}q^{\sigma}\sigma_4(m)+\eta^{\mu\rho}q^{\nu}q^{\sigma}\sigma_5(m)+\eta^{\mu\sigma}q^{\nu}q^{\rho}\sigma_6(m)+\eta^{\nu\rho}q^{\mu}q^{\sigma}\sigma_7(m)+\\
&\eta^{\nu\sigma}q^{\mu}q^{\rho}\sigma_8(m)+\eta^{\rho\sigma}q^{\mu}q^{\nu}\sigma_9(m)+q^{\mu}q^{\nu}q^{\rho}q^{\sigma}\sigma_{10}(m))\,\Omega_m(dq)\,dm\\
\approx&\int_{m=0}^{\infty}\int_{{\vct q}\in{\bf R}^3}\chi_{(a,b)}(\omega_m({\vct q}))\chi_{B_{\epsilon}({\vct0})}({\vct q})(q^2\eta^{\mu\nu}\eta^{\rho\sigma}\sigma_1(m)+q^2\eta^{\mu\rho}\eta^{\nu\sigma}\sigma_2(m)+\\
&q^2\eta^{\mu\sigma}\eta^{\nu\rho}\sigma_3(m)+\eta^{\mu\nu}q^{\rho}q^{\sigma}\sigma_4(m)+\eta^{\mu\rho}q^{\nu}q^{\sigma}\sigma_5(m)+\eta^{\mu\sigma}q^{\nu}q^{\rho}\sigma_6(m)+\\
&\eta^{\nu\rho}q^{\mu}q^{\sigma}\sigma_7(m)+\eta^{\nu\sigma}q^{\mu}q^{\rho}\sigma_8(m)+\eta^{\rho\sigma}q^{\mu}q^{\nu}\sigma_9(m)+\\
&q^{\mu}q^{\nu}q^{\rho}q^{\sigma}\sigma_{10}(m))\omega_m({\vct q})^{-1}\,d{\vct q}\,dm,
\end{align*}
where $q=(\omega_m({\vct q}),{\vct q})$. Therefore, since $\omega_m({\vct0})=m$, $(\omega_m({\vct0}),{\vct0})^2=m^2$ and $(\omega_m({\vct0}),{\vct0})^{\mu}=m\eta^{0\mu}$, we have
\begin{align*}
g^{\mu\nu\rho\sigma}(a,b,\epsilon)\approx&\int_{m=a}^b\int_{{\vct q}\in B_{\epsilon}({\vct0})}(m^2\eta^{\mu\nu}\eta^{\rho\sigma}\sigma_1(m)+m^2\eta^{\mu\rho}\eta^{\nu\sigma}\sigma_2(m)+m^2\eta^{\mu\sigma}\eta^{\nu\rho}\sigma_3(m)+\\
&\eta^{\mu\nu}m^2\eta^{0\rho}\eta^{0\sigma}\sigma_4(m)+\eta^{\mu\rho}m^2\eta^{0\nu}\eta^{0\sigma}\sigma_5(m)+\eta^{\mu\sigma}m^2\eta^{0\nu}\eta^{0\rho}\sigma_6(m)+\\
&\eta^{\nu\rho}m^2\eta^{0\mu}\eta^{0\sigma}\sigma_7(m)+\eta^{\nu\sigma}m^2\eta^{0\mu}\eta^{0\rho}\sigma_8(m)+\eta^{\rho\sigma}m^2\eta^{0\mu}\eta^{0\nu}\sigma_9(m)+\\
&m^4\eta^{0\mu}\eta^{0\nu}\eta^{0\rho}\eta^{0\sigma}\sigma_{10}(m))m^{-1}\,d{\vct q}\,dm\\
\approx&\frac{4\pi}{3}\epsilon^3\int_{m=a}^bm[\eta^{\mu\nu}\eta^{\rho\sigma}\sigma_1(m)+\eta^{\mu\rho}\eta^{\nu\sigma}\sigma_2(m)+\eta^{\mu\sigma}\eta^{\nu\rho}\sigma_3(m)+\\
&\eta^{\mu\nu}\eta^{0\rho}\eta^{0\sigma}\sigma_4(m)+\eta^{\mu\rho}\eta^{0\nu}\eta^{0\sigma}\sigma_5(m)+\eta^{\mu\sigma}\eta^{0\nu}\eta^{0\rho}\sigma_6(m)+\\
&\eta^{\nu\rho}\eta^{0\mu}\eta^{0\sigma}\sigma_7(m)+\eta^{\nu\sigma}\eta^{0\mu}\eta^{0\rho}\sigma_8(m)+\eta^{\rho\sigma}\eta^{0\mu}\eta^{0\nu}\sigma_9(m)+\\
&m^2\eta^{0\mu}\eta^{0\nu}\eta^{0\rho}\eta^{0\sigma}\sigma_{10}(m)]\,dm.
\end{align*}
We define the functions $g^{\mu\nu\rho\sigma}_a(b)$ by
\[ g^{\mu\nu\rho\sigma}_a(b)=\lim_{\epsilon\rightarrow0}(\epsilon^{-3}g^{\mu\nu\rho\sigma}(a,b,\epsilon)). \]
Note that the symbol $\approx$ in the above calculations indicates equality in the limit as $\epsilon\rightarrow0$, this is proved rigorously in \cite{Symmetry}. Therefore
\begin{align*}
g^{\mu\nu\rho\sigma}_a(b)=&\lim_{\epsilon\rightarrow0}(\epsilon^{-3}g^{\mu\nu\rho\sigma}(a,b,\epsilon))\\
=&\frac{4\pi}{3}\int_{m=a}^bm[\eta^{\mu\nu}\eta^{\rho\sigma}\sigma_1(m)+\eta^{\mu\rho}\eta^{\nu\sigma}\sigma_2(m)+\eta^{\mu\sigma}\eta^{\nu\rho}\sigma_3(m)+\\
&\eta^{\mu\nu}\eta^{0\rho}\eta^{0\sigma}\sigma_4(m)+\eta^{\mu\rho}\eta^{0\nu}\eta^{0\sigma}\sigma_5(m)+\eta^{\mu\sigma}\eta^{0\nu}\eta^{0\rho}\sigma_6(m)+\eta^{\nu\rho}\eta^{0\mu}\eta^{0\sigma}\sigma_7(m)+\\
&\eta^{\nu\sigma}\eta^{0\mu}\eta^{0\rho}\sigma_8(m)+\eta^{\rho\sigma}\eta^{0\mu}\eta^{0\nu}\sigma_9(m)+m^2\eta^{0\mu}\eta^{0\nu}\eta^{0\rho}\eta^{0\sigma}\sigma_{10}(m)]\,dm.
\end{align*}
Thus, by the first fundamental theorem of calculus,
\begin{align*}
g^{\mu\nu\rho\sigma\prime}_a(b)=&\frac{4\pi}{3}b[\eta^{\mu\nu}\eta^{\rho\sigma}\sigma_1(b)+\eta^{\mu\rho}\eta^{\nu\sigma}\sigma_2(b)+\eta^{\mu\sigma}\eta^{\nu\rho}\sigma_3(b)+\\
&\eta^{\mu\nu}\eta^{0\rho}\eta^{0\sigma}\sigma_4(b)+\eta^{\mu\rho}\eta^{0\nu}\eta^{0\sigma}\sigma_5(b)+\eta^{\mu\sigma}\eta^{0\nu}\eta^{0\rho}\sigma_6(b)+\eta^{\nu\rho}\eta^{0\mu}\eta^{0\sigma}\sigma_7(b)+\\
&\eta^{\nu\sigma}\eta^{0\mu}\eta^{0\rho}\sigma_8(b)+\eta^{\rho\sigma}\eta^{0\mu}\eta^{0\nu}\sigma_9(b)+b^2\eta^{0\mu}\eta^{0\nu}\eta^{0\rho}\eta^{0\sigma}\sigma_{10}(b)].
\end{align*}
We will now compute the spectral functions $\sigma_i$ from the functions $g^{\mu\nu\rho\sigma\prime}_a(b)$ in a similar manner to that which was used in the case of Lorentz covariant 2-tensor valued measures on Minkowski space \cite{Symmetry,IJMPA,NPB} but the computation is a bit more complicated.
Let $a=a(\mu,\nu,\rho,\sigma)$ be the function of $b$ defined by $(a(\mu,\nu,\rho,\sigma))(b)=\frac{3}{4\pi b}g_a^{\mu\nu\rho\sigma\prime}(b)$ and                                                                                                      
\begin{align*}
X=(X(\mu,\nu,\rho,\sigma))(b)=&\eta^{\mu\nu}\eta^{\rho\sigma}\sigma_1(b)+\eta^{\mu\rho}\eta^{\nu\sigma}\sigma_2(b)+\eta^{\mu\sigma}\eta^{\nu\rho}\sigma_3(b)+\\
&\eta^{\mu\nu}\eta^{0\rho}\eta^{0\sigma}\sigma_4(b)+\eta^{\mu\rho}\eta^{0\nu}\eta^{0\sigma}\sigma_5(b)+\eta^{\mu\sigma}\eta^{0\nu}\eta^{0\rho}\sigma_6(b)+\\
&\eta^{\nu\rho}\eta^{0\mu}\eta^{0\sigma}\sigma_7(b)+\eta^{\nu\sigma}\eta^{0\mu}\eta^{0\rho}\sigma_8(b)+\eta^{\rho\sigma}\eta^{0\mu}\eta^{0\nu}\sigma_9(b)+\\
&b^2\eta^{0\mu}\eta^{0\nu}\eta^{0\rho}\eta^{0\sigma}\sigma_{10}(b).
\end{align*}
Then we have the following cases.\newpage
\begin{center}
{\bf Table 1.: The quantity $a(\mu,\nu,\rho,\sigma)=\frac{3}{4\pi b}g_a^{\mu\nu\rho\sigma\prime}(b)$}\\
\begin{tabular}{||ccc||} \hline
$(\mu,\nu,\rho,\sigma)$&$a(\mu,\nu,\rho,\sigma)$&$X$\\ \hline
(0,0,0,0)&$a_1$&$\sum_{i=1}^9\sigma_i(b)+b^2\sigma_{10}(b)$\\
(0,0,0,1)&$$&0\\
(0,0,1,0)&$$&0\\
(0,0,1,1)&$a_2$&$-\sigma_1(b)-\sigma_9(b)$\\
(0,1,0,0)&$$&0\\
(0,1,0,1)&$a_3$&$-\sigma_2(b)-\sigma_8(b)$\\
(0,1,1,0)&$a_4$&$-\sigma_3(b)-\sigma_7(b)$\\
(0,1,1,1)&$$&0\\
(1,0,0,0)&$$&0\\
(1,0,0,1)&$a_{5}$&$-\sigma_3(b)-\sigma_6(b)$\\
(1,0,1,0)&$a_{6}$&$-\sigma_2(b)-\sigma_5(b)$\\
(1,0,1,1)&$$&0\\
(1,1,0,0)&$a_{7}$&$-\sigma_1(b)-\sigma_4(b)$\\
(1,1,0,1)&$$&0\\
(1,1,1,0)&$$&0\\
(1,1,1,1)&$a_{8}$&$\sigma_1(b)+\sigma_2(b)+\sigma_3(b)$\\
\hline \hline
\end{tabular}
\end{center}
The $a(\mu,\nu,\rho,\sigma)$ are known functions. We would like to determine $\sigma_1(b),\ldots,\sigma_{10}(b)$. 
There are only 8 equations for 10 unknowns. Therefore the solution space will be a hyperplane parameterized by a number of parameters. Since $\sigma_{10}(m)$ only occurs in one place we will take $c=b^2\sigma_{10}(m)$ to be one of the parameters. Rewrite $a_1\rightarrow a_1-c$ i.e. absorb $-c$ into the $a_1$ for the moment. Then the first row of Table 1 becomes
\begin{center}
\begin{tabular}{||ccc||} \hline
$(\mu,\nu,\rho,\sigma)$&$a$&$X$\\ \hline
(0,0,0,0)&$a_1$&$\sum_{i=1}^9\sigma_i(b)$\\ 
\hline \hline
\end{tabular}
\end{center}
Then it is easy to see that if spectral functions $\sigma_i$ satisfy the equations of Table 1 then $\sum_{i=1}^8a_i=0$. Omitting the arguments on the $\sigma_i$ functions and the $a_j$ functions we have that the equations of Table 1 are equivalent to the following matrix equation 
\[ \left(\begin{array}{cccccccccc}
1&1&1&1&1&1&1&1&1\\
-1&0&0&0&0&0&0&0&-1\\
0&-1&0&0&0&0&0&-1&0\\
0&0&-1&0&0&0&-1&0&0\\
0&0&-1&0&0&-1&0&0&0\\
0&-1&0&0&-1&0&0&0&0\\
-1&0&0&-1&0&0&0&0&0\\
1&1&1&0&0&0&0&0&0
\end{array}\right)\left(\begin{array}{c}\sigma_1\\\sigma_2\\\sigma_3\\\sigma_4\\\sigma_5\\\sigma_6\\\sigma_7\\\sigma_8\\\sigma_9\end{array}\right)=\left(\begin{array}{c}a_1\\a_2\\a_3\\a_4\\a_5\\a_6\\a_7\\a_8\end{array}\right). \]
Carrying out row-echelon reduction we obtain the equivalent matrix equation
\begin{equation}\label{eq:sigma_solution}
\left(\begin{array}{cccccccccc}
1&0&0&0&0&0&0&0&1\\
0&1&0&0&0&0&0&1&0\\
0&0&1&0&0&0&0&-1&-1\\
0&0&0&1&0&0&0&0&-1\\
0&0&0&0&1&0&0&-1&0\\
0&0&0&0&0&1&0&1&1\\
0&0&0&0&0&0&1&1&1\\
0&0&0&0&0&0&0&0&0
\end{array}\right)\left(\begin{array}{c}\sigma_1\\\sigma_2\\\sigma_3\\\sigma_4\\\sigma_5\\\sigma_6\\\sigma_7\\\sigma_8\\\sigma_9\end{array}\right)=\left(\begin{array}{c}-a_2\\-a_3\\-a_1-a_4-a_5-a_6-a_7\\a_2-a_7\\a_3-a_6\\a_1+a_4+a_6+a_7\\a_1+a_5+a_6+a_7\\\sum_{i=1}^8a_i\end{array}\right).
\end{equation}
Let $c_1=\sigma_8,c_2=\sigma_9$. Then the general solution set to the set of equations is
\begin{align*}
\sigma_1=&-a_2-c_2,\\
\sigma_2=&-a_3-c_1,\\
\sigma_3=&-a_1-a_4-a_5-a_6-a_7+c_1+c_2,\\
\sigma_4=&a_2-a_7+c_2,\\
\sigma_5=&a_3-a_6+c_1,\\
\sigma_6=&a_1+a_4+a_6+a_7-c_1-c_2,\\
\sigma_7=&a_1+a_5+a_6+a_7-c_1-c_2,\\
\sigma_8=&c_1,\\
\sigma_9=&c_2,
\end{align*}
as long as $\sum_{i=1}^8a_i=0$.
Now we need to make the substitution $a_1\rightarrow a_1-c$ to express the solution in terms of the untransformed value of $a_1$. Thus the general solution is given by
\begin{align}\label{eq:sigma_solution_1}
\sigma_1=&-a_2-c_2,\nonumber\\
\sigma_2=&-a_3-c_1,\nonumber\\
\sigma_3=&-a_1+c-a_4-a_5-a_6-a_7+c_1+c_2,\nonumber\\
\sigma_4=&a_2-a_7+c_2,\nonumber\\
\sigma_5=&a_3-a_6+c_1,\\
\sigma_6=&a_1-c+a_4+a_6+a_7-c_1-c_2,\nonumber\\
\sigma_7=&a_1-c+a_5+a_6+a_7-c_1-c_2,\nonumber\\
\sigma_8=&c_1,\nonumber\\
\sigma_9=&c_2.\nonumber\\
\sigma_{10}=&b^{-2}c\nonumber,
\end{align}
with $c=\sum_{i=1}^8a_i$.
This general solution represents a 2 dimensional hyperplane in ${\bf C}^{10}$. Thus the spectral calculus for 4-tensor valued measures can only pin down the solution to within a 2-plane in ${\bf C}^{10}$ in the case of general causal Lorentz covariant Borel complex 4-tensor valued measures  on Minkowski space.

\subsection{Unique solution in the case when $\Pi^{\mu\nu\rho\sigma}$ has certain symmetry properties}

In the case when $\Pi^{\mu\nu\rho\sigma}$ has symmetry properties in addition to Lorentz covariance we can determine additional equations that must be satisfied by the spectral functions $\sigma_i$.
Consider the following symmetry properties.
\begin{align}\label{eq:Pi_tensor_symmetry}
&\Pi^{\mu\nu\rho\sigma}=\Pi^{\nu\mu\rho\sigma},\forall\mu,\nu,\rho,\sigma\in\{0,1,2,3\},\nonumber\\ 
&\Pi^{\mu\nu\rho\sigma}=\Pi^{\mu\nu\sigma\rho},\forall\mu,\nu,\rho,\sigma\in\{0,1,2,3\},\\
&\Pi^{\mu\nu\rho\sigma}=\Pi^{\rho\sigma\mu\nu},\forall\mu,\nu,\rho,\sigma\in\{0,1,2,3\}.\nonumber
\end{align}
First note that if the tensor $\Pi^{\mu\nu\rho\sigma}$ has these symmetry properties then the tensor $g_a^{\mu\nu\rho\sigma\prime}$ has these symmetry properties. Therefore, referring to Table 1, we must have that
\begin{align*}
&a_2=a_7,\\
&a_3=a_4=a_5=a_6.
\end{align*}
Let
\begin{align*}
&b_1=a_1,\\
&b_2=a_2=a_7,\\
&b_3=a_3=a_4=a_5=a_6,\\
&b_4=a_8,
\end{align*}
Therefore Eqns.~\ref{eq:sigma_solution_1} becomes
\begin{align}\label{eq:sigma_solution_2}
\sigma_1=&-b_2-c_2,\nonumber\\
\sigma_2=&-b_3-c_1,\nonumber\\
\sigma_3=&b_2+b_3+b_4+c_1+c_2,\nonumber\\
\sigma_4=&c_2,\nonumber\\
\sigma_5=&c_1,\\
\sigma_6=&-b_2-2b_3-b_4-c_1-c_2,\nonumber\\
\sigma_7=&-b_2-2b_3-b_4-c_1-c_2,\nonumber\\
\sigma_8=&c_1,\nonumber\\
\sigma_9=&c_2.\nonumber\\
\sigma_{10}=&b^{-2}\sum_{i=1}^8a_i=b^{-2}(b_1+2b_2+4b_3+b_4)\nonumber.
\end{align}
Now we have the following
\begin{lemma} $\Pi^{\mu\nu\rho\sigma}$ is symmetric in $(\mu,\nu)$ if and only if 
\begin{equation}\label{eq:symmetry_condition}
\sigma_3=\sigma_2,\sigma_7=\sigma_5\text{ and }\sigma_8=\sigma_6.
\end{equation}
\end{lemma}
{\noindent{\bf Proof}}
If the density $\Pi^{\mu\nu\rho\sigma}$ has the given symmetry property then the measure $\Pi^{\mu\nu\rho\sigma}$ has the property. For the converse if the measure has the property then the density has the property almost everywhere in which case it can be adjusted to have the symmetry property for all $q\in{\bf R}^4$. Now, by Eq.~\ref{eq:canonical_form_density} the density has the symmetry property if and only if
\begin{align}\label{eq:density_symmetry_property}
0=&\Pi^{\mu\nu\rho\sigma}(q)-\Pi^{\nu\mu\rho\sigma}(q)\nonumber\\
=&Q[\eta^{\mu\rho}\eta^{\nu\sigma}(\sigma_2(Q)-\sigma_3(Q))+\eta^{\mu\sigma}\eta^{\nu\rho}(\sigma_3(Q)-\sigma_2(Q))]+\nonumber\\
&Q^{-1}[\eta^{\mu\rho}q^{\nu}q^{\sigma}(\sigma_5(Q)-\sigma_7(Q))+\eta^{\nu\rho}q^{\mu}q^{\sigma}(\sigma_7(Q)-\sigma_5(Q))+\nonumber\\
&\eta^{\mu\sigma}q^{\nu}q^{\rho}(\sigma_6(Q)-\sigma_8(Q))+\eta^{\nu\sigma}q^{\mu}q^{\rho}(\sigma_8(Q)-\sigma_6(Q))],
\end{align} 
for all $\mu,\nu,\rho,\sigma\in\{0,1,2,3\},q\in{\bf R}^4$ s.t. $q^2=Q^2,q^0>0$ with $Q>0$. Therefore Eq.~\ref{eq:symmetry_condition} is sufficient for given symmetry property. 

Now, for the converse, suppose that Eq.~\ref{eq:density_symmetry_property} is satisfied. Then
\begin{align*}
0=&Q[(\eta^{\mu\rho}\eta^{\nu\sigma}-\eta^{\mu\sigma}\eta^{\nu\rho})(\sigma_2(Q)-\sigma_3(Q)+\\
&Q^{-1}[(\eta^{\mu\rho}q^{\nu}q^{\sigma}-\eta^{\nu\rho}q^{\mu}q^{\sigma})(\sigma_5(Q)-\sigma_7(Q))+\\
&(\eta^{\mu\sigma}q^{\nu}q^{\rho}-\eta^{\nu\sigma}q^{\mu}q^{\rho})(\sigma_6(Q)-\sigma_8(Q)].
\end{align*}
Taking $\mu=0,\rho=0,\nu=1,\sigma=2$ we have $Q^{-1}q^1q^2(\sigma_7(Q)-\sigma_5(Q))=0$   for all $q\in{\bf R}^4$ such that $q=(q^0,q^1,q^2,q^3),q^2=Q^2,q^0>0$. It follows that
\[ \sigma_7(Q)=\sigma_5(Q),\forall Q>0. \]
By making other choices the other two required equalities can be proved.
 $\Box$

By similar calculations we can show that $\Pi^{\mu\nu\rho\sigma}$ is symmetric in $(\rho,\sigma)$ if and only if $\sigma_3=\sigma_2,\sigma_6=\sigma_5$ and $\sigma_8=\sigma_7$ and furthermore that $\Pi^{\mu\nu\rho\sigma}$ is symmetric under the transformation $(\mu,\nu)\leftrightarrow(\rho,\sigma)$ if and only if $\sigma_9=\sigma_4$ and $\sigma_7=\sigma_6$.
 
Thus, if $\Pi^{\mu\nu\rho\sigma}$ has the symmetry properties given by Eq.~\ref{eq:Pi_tensor_symmetry} then
\begin{align*}
&\sigma_2=\sigma_3,\\
&\sigma_5=\sigma_6=\sigma_7=\sigma_8,\\
&\sigma_4=\sigma_9.
\end{align*}
Let
\begin{align*}
&\rho_1=\sigma_1,\\
&\rho_2=\sigma_2=\sigma_3,\\
&\rho_3=\sigma_4=\sigma_9,\\
&\rho_4=\sigma_5=\sigma_6=\sigma_7=\sigma_8.
\end{align*}
Then, the system Eq.~\ref{eq:sigma_solution_2} of equations is equivalent to the equations
\begin{align*}
&\rho_1=-b_2-\rho_3,\\
&\rho_2=-b_3-\rho_4,\\
&\rho_2=b_2+b_3+b_4+\rho_4+\rho_3,\\
&\rho_4=-b_2-2b_3-b_4-\rho_4-\rho_3.
\end{align*}
These equations can be written as
\[ \left(\begin{array}{cccc}
1&0&1&0\\
0&1&0&1\\
0&1&-1&-1\\
0&0&1&2
\end{array}\right)\left(\begin{array}{c}
\rho_1\\
\rho_2\\
\rho_3\\
\rho_4
\end{array}\right)=\left(\begin{array}{c}
-b_2\\
-b_3\\
b_2+b_3+b_4\\
-b_2-2b_3-b_4
\end{array}\right). \]
On carrying out row-echelon reduction this system of equations is equivalent to
\[ \left(\begin{array}{cccc}
1&0&0&-2\\
0&1&0&1\\
0&0&1&2\\
0&0&0&0
\end{array}\right)\left(\begin{array}{c}
\rho_1\\
\rho_2\\
\rho_3\\
\rho_4
\end{array}\right)=\left(\begin{array}{c}
2b_3+b_4\\
-b_3\\
-b_2-2b_3-b_4\\
0
\end{array}\right). \]
Hence the system of equations is feasible but it is underdetermined. We need another independent equation satisfied by the spectral functions $\rho_i$ in order to have a system with a unique solution. To do this we contract the tensor $\Pi^{\mu\nu\rho\sigma}$ with respect to its last two indices as follows
\[ \Sigma^{\mu\nu}=\eta_{\rho\sigma}\Pi^{\mu\nu\rho\sigma}. \]
From Eq.~\ref{eq:canonical_form_density} it is straightforward to compute that
\begin{align*}
\Sigma^{\mu\nu}(q)=&Q\eta^{\mu\nu}(4\sigma_1+\sigma_2+\sigma_3)+Q^{-1}(Q^2\eta^{\mu\nu}\sigma_4+q^{\mu}q^{\nu}(\sigma_5+\sigma_6+\sigma_7+\sigma_8+4\sigma_9)+\\
&Q^2q^{\mu}q^{\nu}\sigma_{10})\\
=&Q\eta^{\mu\nu}(4\sigma_1+\sigma_2+\sigma_3+\sigma_4)+Q^{-1}q^{\mu}q^{\nu}(\sigma_5+\sigma_6+\sigma_7+\sigma_8+4\sigma_9+Q^2\sigma_{10})\\
=&Q\eta^{\mu\nu}(4\rho_1+2\rho_2+\rho_3)+Q^{-1}q^{\mu}q^{\nu}(4\rho_4+4\rho_3+b_1+2b_2+4b_3+b_4)\\
=&Q^2\eta^{\mu\nu}\zeta_1(Q)+q^{\mu}q^{\nu}\zeta_2(Q),
\end{align*}
where
\begin{align*}
\zeta_1(Q)=&Q^{-1}(4\rho_1+2\rho_2+\rho_3),\\
\zeta_2(Q)=&Q^{-1}(4(\rho_3+\rho_4)+b_1+2b_2+4b_3+b_4).
\end{align*}
$\Sigma^{\mu\nu}$ is the density for a measure $\Gamma\mapsto\Sigma^{\mu\nu}(\Gamma)$ which is a causal Lorentz covariant Borel 2-tensor valued measure. We may now carry out the spectral calculus for 2-tensor measures on Minkowski space \cite{NPB} to obtain additional equations for the spectral functions $\rho_i$. Specifically, let $h(a,b,\epsilon)=\Sigma(\Gamma(a,b,\epsilon))$ and then compute $h_a^{\mu\nu}(b)=\lim_{\epsilon\rightarrow0}(\epsilon^{-3}h^{\mu\nu}(a,b,\epsilon))$. Then the spectral functions $\zeta_1$ and $\zeta_2$ can be computed from the formulas \cite{NPB}
\begin{align*}
\zeta_1(Q)=&-\frac{3}{4\pi}\frac{1}{Q}h_a^{ii\prime}(Q),i=1,2,3,\\
\zeta_2(Q)=&\frac{3}{4\pi}\frac{1}{Q}h_a^{00\prime}(Q)-\zeta_1(Q).
\end{align*}  
It is straightforward to show that the $h_a^{\mu\nu}$ functions can be computed from the formula
\begin{equation}\label{eq:compute_h_a}
h_a^{\mu\nu}(b)=\eta_{\rho\sigma}g_a^{\mu\nu\rho\sigma}(b),\forall b>0.
\end{equation} 
Letting
\begin{align}\label{eq:xi_i_solution}
\xi_1=&Q\zeta_1,\nonumber\\
\xi_2=&Q\zeta_2-(b_1+2b_2+4b_3+b_4),
\end{align}
we now have the following enlarged system of equations.
\[ \left(\begin{array}{cccc}
1&0&0&-2\\
0&1&0&1\\
0&0&1&2\\
4&2&1&0\\
0&0&4&4
\end{array}\right)\left(\begin{array}{c}
\rho_1\\
\rho_2\\
\rho_3\\
\rho_4
\end{array}\right)=\left(\begin{array}{c}
2b_3+b_4\\
-b_3\\
-b_2-2b_3-b_4\\
\xi_1\\
\xi_2
\end{array}\right). \]
On carrying out row-echelon reduction this system of equations is equivalent to
\begin{equation}\label{eq:rho_i_solution}
\left(\begin{array}{cccc}
1&0&0&0\\
0&1&0&0\\
0&0&1&0\\
0&0&0&1\\
0&0&0&1
\end{array}\right)\left(\begin{array}{c}
\rho_1\\
\rho_2\\
\rho_3\\
\rho_4
\end{array}\right)=\left(\begin{array}{c}
-\frac{1}{2}\xi_2-2b_2-2b_3-b_4\\
\frac{1}{4}\xi_2+b_2+b_3+b_4\\
\frac{1}{2}\xi_2+b_2+2b_3+b_4\\
\frac{1}{4}(\xi_1+b_2-4b_3-3b_4)\\
-\frac{1}{4}\xi_2-b_2-2b_3-b_4
\end{array}\right).
\end{equation}
These equations are feasible if and only if $\xi_1$ and $\xi_2$ satisfy the consistency condition
\[ \frac{1}{4}(\xi_1+b_2-4b_3-3b_4)=-\frac{1}{4}\xi_2-b_2-2b_3-b_4, \]
which is equivalent to the condition
\begin{equation}
\xi_1+\xi_2=-5b_2-4b_3-b_4,
\end{equation}
from which it follows that the system is consistent if and only if
\begin{equation}\label{eq:consistency_condition}
\frac{3}{4\pi Q}h_a^{00\prime}(Q)=(b_1(Q)-3b_2(Q)),\forall Q>0.
\end{equation}
However we can easily prove the following
\begin{lemma}
Let $\Pi^{\mu\nu\rho\sigma}$ be causal Lorentz covariant Borel 4-tensor valued measure on Minkowski space and define $\Sigma^{\mu\nu}=\eta_{\rho\sigma}\Pi^{\mu\nu\rho\sigma}$. Then if the spectral calculus is applied to $\Sigma^{\mu\nu}$ resulting in a function $h_a^{00\prime}$ then the consistency condition given by Eq.~\ref{eq:consistency_condition} is satisfied.
\end{lemma}
{\noindent{\bf Proof}} We have
\[ h_a^{\mu\nu\prime}=\eta_{\rho\sigma}g_a^{\mu\nu\rho\sigma\prime}. \]
Therefore
\begin{align*}
h_a^{00\prime}(Q)=&\eta_{\rho\sigma}g^{00\rho\sigma\prime}(Q)=g_a^{0000\prime}(Q)-\sum_{i=1}^3g_a^{00ii\prime}(Q)=\frac{4\pi}{3}Q(a_1(Q)-3a_2(Q))\\
=&\frac{4\pi}{3}Q(b_1(Q)-3b_2(Q)),\forall Q>0.
\end{align*}
$\Box$

Therefore, assuming that the measure $\Pi^{\mu\nu\rho\sigma}$ satisfies the symmetry conditions given by Eq.~\ref{eq:Pi_tensor_symmetry} it and its density can be completely and uniquely  determined using the spectral calculus. 

\section{Summary of spectral calculus algorithm in the presence of symmetries\label{section:spectral_calculus_summary}}

The computations of the last section may seem complicated and intricate. However they can be summarized resulting in a very simple algorithm for computing the function (density) associated with any causal Lorentz covariant Borel complex 4-tensor valued measure $\Pi^{\mu\nu\rho\sigma}$ on Minkowski space which satisfies the symmetry condition given by Eq.~\ref{eq:Pi_tensor_symmetry}.

Having computed the functions $g^{\mu\nu\rho\sigma\prime}_a$ for $\Pi^{\mu\nu\rho\sigma}$   (which we will do below for the case of the graviton vacuum polarization tensor) compute
\[ h_a^{\mu\nu\prime}=\eta_{\rho\sigma}g^{\mu\nu\rho\sigma\prime}, \]
and
\begin{align*}
&b_1=g_a^{0000\prime},\\
&b_2=g_a^{0011\prime},\\
&b_3=g_a^{0101\prime},\\
&b_4=g_a^{1111\prime},
\end{align*} 
and the $\xi_i$ functions using Eq.~\ref{eq:xi_i_solution}.
Then determine the four spectral functions $\rho_i$ by plugging the functions $\xi_1,\xi_2,b_2,b_3$ and $b_4$ into the equation Eq.~\ref{eq:rho_i_solution}. Then compute the 10 spectral functions $\sigma_i$ using the equations
\begin{align*}
&\sigma_1=\rho_1,\\
&\sigma_2=\sigma_3=\rho_2,\\
&\sigma_4=\sigma_9=\rho_3,\\
&\sigma_5=\sigma_6=\sigma_7=\sigma_8=\rho_4\\
&\sigma_{10}=b^{-2}(b_1+2b_2+4b_3+b_4).
\end{align*}
Finally, plug the spectral functions $\sigma_i$ into the equation Eq.~\ref{eq:canonical_form_density} to obtain the density $\Pi^{\mu\nu\rho\sigma}$.

\section{Lorentz covariance of objects occurring in QG}

There are two basic types of (mathematical) objects occurring in all formulations of pQG (perturbative QG) and also non-perturbative approaches to QG. These are the propagators and the vertices. QG has an infinite sequence of vertex types of higher and higher order. These objects are all Lorentz covariant in all modern formulations of QG.

In the  De Donder harmonic (Feynman) gauge the QG graviton propagator is given by \cite{Duff}
\begin{align*}
D^{\mu\nu\rho\sigma}(q)=&\frac{1}{q^2+i\epsilon}(\eta^{\mu\rho}\eta^{\nu\sigma}+\eta^{\mu\sigma}\eta^{\nu\rho}-\eta^{\mu\nu}\eta^{\rho\sigma})\\
=&\frac{1}{q^2+i\epsilon}P^{\mu\nu\rho\sigma},
\end{align*}
say. This is clearly in the form of the canonical form. Therefore it is Lorentz covariant.

The graviton vertices in classical QG are obtained from the Taylor expansion of the Einstein-Hilbert Lagrangian which is constructed from a curvature scalar. Such a scalar is generally covariant and therefore, in particular, Lorentz covariant. Tensors derived from it are Lorentz covariant. Hence the graviton vertices are Lorentz covariant \cite{DewittII}.

In other, more modern approaches to QG such as spin foam loop quantum gravity or string theory the graviton vertices are explicitly designed and constructed to be Lorentz covariant.

In particular, the 3-graviton vertex $V^{(3)}_{\alpha_1\beta_1\alpha_2\beta_2\alpha_3\beta_3}$ is Lorentz covariant. Explicit expressions for the classical form for $V^{(3)}$ can be found in the references \cite{DewittIII,Duff}. The components of $V^{(3)}$ are polynomial functions of its momentum arguments $p_1,p_2,p_3$, of order $\leq2$. $V^{(3)}$ has certain symmetry properties \cite{Duff} which we will discuss below.

\section{Representation of the graviton vacuum polarization tensor $\Pi^{\mu\nu\rho\sigma}$ using CSR}

We follow the convention of Duff \cite{Duff} given by Fig. \ref{fig:convention}.

\begin{figure} 
\centering
\includegraphics[width=8cm]{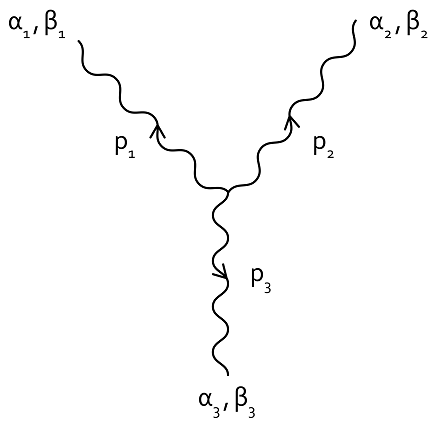}
\caption{Convention for graviton Feynman diagrams, example of the $V^{(3)}$ 3-graviton vertex  \label{fig:convention}}
\end{figure}
 
Applying the Feynman rules the graviton vacuum polarization function which is associated with the Feynman diagram of Fig. \ref{fig:VP} is given by
\[ \Pi^{\mu\nu\rho\sigma}=\text{Trace}\int(\text{vertex$^{\mu\nu}\times$propagator$\times$vertex$^{\rho\sigma}\times$propagator)}, \]
i.e., in detail,
\begin{align}\label{eq:VP_Feynman_integral}
\Pi^{\mu\nu\rho\sigma}(q)=&\int\frac{dk}{(2\pi)^4}i\kappa V^{(3)\mu\nu}_{\alpha_1\beta_1\alpha_2\beta_2}(-q,q-k,k)iD^{\alpha_2\beta_2\alpha_3\beta_3}(k-q)\nonumber\\
&i\kappa V^{(3)\rho\sigma}_{\alpha_3\beta_3\alpha_4\beta_4}(k-q,q,-k)iD^{\alpha_4\beta_4\alpha_1\beta_1}(k)\nonumber\\
=&\frac{\kappa^2}{(2\pi)^4}\int V^{(3)\mu\nu}_{\alpha_1\beta_1\alpha_2\beta_2}(-q,q-k,k)D^{\alpha_2\beta_2\alpha_3\beta_3}(k-q)\nonumber\\
&V^{(3)\rho\sigma}_{\alpha_3\beta_3\alpha_4\beta_4}(k-q,q,-k)D^{\alpha_4\beta_4\alpha_1\beta_1}(k)\,dk.
\end{align}

\begin{figure} 
\centering
\includegraphics[width=11cm]{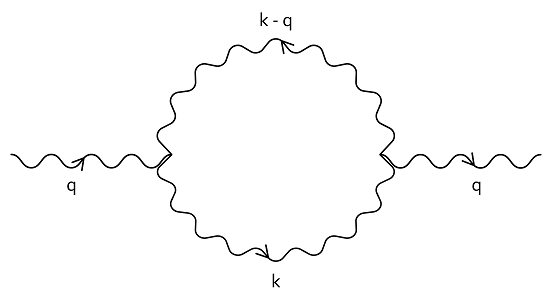}
\caption{Feynman diagram for graviton vacuum polarization (VP)\label{fig:VP}}
\end{figure}

Unfortunately, as is well known, this integral is not convergent. This has been an issue for about half a century and the main approach to dealing with this problem has been to use classical, renormalization based, techniques for circumventing divergences. The problem with this approach is that QG involves a multitude of different graviton vertices of higher and higher orders and thus to solve the problem using renormalization based methods would require an infinite number of counterterms which renders the theory infeasible and physically meaningless.

Our approach to the problem is to use CSR in which a divergent Feynman integral is represented by a Lorentz covariant measure. We then use the spectral calculus described above to compute the density for this measure. The density is well defined and finite and can be used in place of $\Pi^{\mu\nu\rho\sigma}$ in all calculations involving this function. 

There is, using this method, no counterterms and the function and its higher order counterparts can be used an arbitrary number of times in QG calculations.

\subsection{Determination of a well defined measure representing $\Pi^{\mu\nu\rho\sigma} $}

We now carry out a formal argument starting with the ``function"  $\Pi^{\mu\nu\rho\sigma}$ and resulting in a well defined Lorentz covariant Borel complex tensor valued measure $\Pi^{\mu\nu\rho\sigma}$ on Minkowski space. As mentioned above, there must be one or more ``formal" non-rigorous steps in such an argument. One can not start with something not well defined (the Feynman integral) and arrive at something well defined (a Lorentz covariant measure) by a sequence of rigorous steps. In the argument below the only non-rigorous step is the interchange at one point in the order of integration which is not universally valid.

We now assume ``pretend" that the function defined by the Feynman integral given by Eq.~\ref{eq:VP_Feynman_integral} exists and compute the measure generated by taking it as a density as follows
\begin{align*}
&\Pi^{\mu\nu\rho\sigma}(\Gamma)=\int_{\Gamma}\Pi^{\mu\nu\rho\sigma}(q)\,dq=\int\chi_{\Gamma}(q)\Pi^{\mu\nu\rho\sigma}(q)\,dq\\
=&\frac{\kappa^2}{(2\pi)^4}\int\chi_{\Gamma}(q)\left[\int V^{(3)\mu\nu}_{\alpha_1\beta_1\alpha_2\beta_2}(-q,q-k,k)D^{\alpha_2\beta_2\alpha_3\beta_3}(k-q)\right.\\
&\left.V^{(3)\rho\sigma}_{\alpha_3\beta_3\alpha_4\beta_4}(k-q,q,-k)D^{\alpha_4\beta_4\alpha_1\beta_1}(k)\,dk\frac{}{}\right]\,dq\\
=&\frac{\kappa^2}{(2\pi)^4}\int\chi_{\Gamma}(q)V^{(3)\mu\nu}_{\alpha_1\beta_1\alpha_2\beta_2}(-q,q-k,k)\frac{1}{(k-q)^2+i\epsilon}P^{\alpha_2\beta_2\alpha_3\beta_3}\\
&V^{(3)\rho\sigma}_{\alpha_3\beta_3\alpha_4\beta_4}(k-q,q,-k)\frac{1}{k^2+i\epsilon}P^{\alpha_4\beta_4\alpha_1\beta_1}\,dk\,dq\\
=&\frac{\kappa^2}{(2\pi)^4}\int\chi_{\Gamma}(q)V^{(3)\mu\nu}_{\alpha_1\beta_1\alpha_2\beta_2}(-q,q-k,k)\frac{1}{(k-q)^2+i\epsilon}P^{\alpha_2\beta_2\alpha_3\beta_3}\\
&V^{(3)\rho\sigma}_{\alpha_3\beta_3\alpha_4\beta_4}(k-q,q,-k)\frac{1}{k^2+i\epsilon}P^{\alpha_4\beta_4\alpha_1\beta_1}\,dq\,dk\\
=&\frac{\kappa^2}{(2\pi)^4}\int\chi_{\Gamma}(q+k)V^{(3)\mu\nu}_{\alpha_1\beta_1\alpha_2\beta_2}(-q-k,q,k)\frac{1}{q^2+i\epsilon}P^{\alpha_2\beta_2\alpha_3\beta_3}\\
&V^{(3)\rho\sigma}_{\alpha_3\beta_3\alpha_4\beta_4}(-q,q+k,-k)\frac{1}{k^2+i\epsilon}P^{\alpha_4\beta_4\alpha_1\beta_1}\,dq\,dk\\
=&-\frac{\kappa^2}{16\pi^2}\int\chi_{\Gamma}(q+k)V^{(3)\mu\nu}_{\alpha_1\beta_1\alpha_2\beta_2}(-q-k,q,k)P^{\alpha_2\beta_2\alpha_3\beta_3}\\
&V^{(3)\rho\sigma}_{\alpha_3\beta_3\alpha_4\beta_4}(-q,q+k,-k)P^{\alpha_4\beta_4\alpha_1\beta_1}\,\Omega_0^{\pm}(dq)\,\Omega_0^{\pm}(dk),
\end{align*}
where in the fourth line we have used the formal substitution $dk\,dq\rightarrow\,dq\,dk$, in the second last line we have used the translation invariance of the Lebesgue measure on ${\bf R}^4$ in the inner integral and in the last line we have used the identification
\[ \frac{1}{q^2-m^2+i\epsilon}\,dq\leftrightarrow-i\pi\,\Omega_m^{\pm}(dq),m\geq0, \]
which is proved in Ref.~\cite{Symmetry}.
Thus (taking the + sign for both the $\Omega_0^{\pm}$) we have that $\Pi^{\mu\nu\rho\sigma}$ is causal and
\begin{align}\label{eq:VP_measure_formula}
\Pi^{\mu\nu\rho\sigma}(\Gamma)=&-\frac{\kappa^2}{16\pi^2}\int\chi_{\Gamma}(q+k)V^{(3)\mu\nu}_{\alpha_1\beta_1\alpha_2\beta_2}(-q-k,q,k)P^{\alpha_2\beta_2\alpha_3\beta_3}\nonumber\\
&V^{(3)\rho\sigma}_{\alpha_3\beta_3\alpha_4\beta_4}(-q,q+k,-k)P^{\alpha_4\beta_4\alpha_1\beta_1}|{\vct q}|^{-1}|{\vct k}|^{-1}\,d{\vct q}\,d{\vct k}\\
=&-\frac{\kappa^2}{16\pi^2}\int_{s_1=0}^{\infty}\int_{\omega_1\in S^2}\int_{s_2=0}^{\infty}\int_{\omega_2\in S^2}\chi_{\Gamma}(q+k)V^{(3)\mu\nu}_{\alpha_1\beta_1\alpha_2\beta_2}(-q-k,q,k)P^{\alpha_2\beta_2\alpha_3\beta_3}\nonumber\\
&V^{(3)\rho\sigma}_{\alpha_3\beta_3\alpha_4\beta_4}(-q,q+k,-k)P^{\alpha_4\beta_4\alpha_1\beta_1}s_1s_2\,d\omega_2\,ds_2\,d\omega_1\,ds_1,\nonumber
\end{align}
in which $S^2$ is the (surface of) the unit sphere in ${\bf R}^3$,  ${\vct q}=s_2\omega_2,{\vct k}=s_1\omega_1,q=(s_2,{\vct q}),k=(s_1,{\vct k})$. 

Clearly $\Pi^{\mu\nu\rho\sigma}(\emptyset)=0$. Suppose that $\Gamma$ is a relatively compact Borel subset of ${\bf R}^4$. Then there exists $R>0$ such that $|q^0|<R,\forall q=(q^0,{\vct q})\in\Gamma$. Then, for all $q=(q^0,{\vct q}),k=(k^0,{\vct k})\in{\bf R}^4$,
\[ \chi_{\Gamma}(q+k)=1\Rightarrow(q+k)\in\Gamma\Rightarrow(|{\vct q}|+|{\vct k}|,{\vct q}+{\vct k})\in\Gamma\Rightarrow|{\vct q}|+|{\vct k}|<R. \]
Thus
\[ |{\vct q}|>R\text{ or }|{\vct k}|>R\Rightarrow\chi_{\Gamma}(q+k)=0. \]Then
\begin{align*}
\Pi^{\mu\nu\rho\sigma}(\Gamma)=&-\frac{\kappa^2}{16\pi^2}\int_{s_1=0}^{R}\int_{\omega_1\in S^2}\int_{s_2=0}^{R}\int_{\omega_2\in S^2}\chi_{\Gamma}(q+k)V^{(3)\mu\nu}_{\alpha_1\beta_1\alpha_2\beta_2}(-q-k,q,k)P^{\alpha_2\beta_2\alpha_3\beta_3}\\
&V^{(3)\rho\sigma}_{\alpha_3\beta_3\alpha_4\beta_4}(-q,q+k,-k)P^{\alpha_4\beta_4\alpha_1\beta_1}s_1s_2\,d\omega_2\,ds_2\,d\omega_1\,ds_1.
\end{align*}
The $P^{\mu\nu\rho\sigma}$ are just numbers and the $V^{(3)\mu\nu}_{\alpha_1\beta_1\alpha_2\beta_2}$ are polynomial functions of their arguments. Therefore $\Pi^{\mu\nu\rho\sigma}(\Gamma)$, being the integral of a locally integrable function over a compact set, is finite. This is true for all relatively compact Borel subsets of Minkowski space. Thus $\Pi^{\mu\nu\rho\sigma}:{\mathcal B}_0({\bf R}^4)\rightarrow{\bf C}$. It is straightforward to show that $\Pi^{\mu\nu\rho\sigma}|_K$ is $\sigma$-additive for all compact $K\subset{\bf R}^4$. Therefore $\Pi^{\mu\nu\rho\sigma}$ is a well defined causal Borel complex tensor valued measure on Minkowski space. 

\subsection{Symmetry properties of $\Pi^{\mu\nu\rho\sigma}$}

It is straightforward, using the Lorentz covariance properties of the 3-graviton vertex $V^{(3)}$ and the tensor $P^{\mu\nu\rho\sigma}$ associated with the propagator $D^{\mu\nu\rho\sigma}$, to prove that $\Pi^{\mu\nu\rho\sigma}$ is Lorentz covariant. 

From Eq.~\ref{eq:VP_measure_formula} $\Pi^{\mu\nu\rho\sigma}$ is given by
\begin{align*}
\Pi^{\mu\nu\rho\sigma}(\Gamma)=&-\frac{\kappa^2}{16\pi^2}\int\chi_{\Gamma}(q+k)V^{(3)\mu\nu}_{\alpha_1\beta_1\alpha_2\beta_2}(-q-k,q,k)P^{\alpha_2\beta_2\alpha_3\beta_3}\\
&V^{(3)\rho\sigma}_{\alpha_3\beta_3\alpha_4\beta_4}(-q,q+k,-k)P^{\alpha_4\beta_4\alpha_1\beta_1}|{\vct q}|^{-1}|{\vct k}|^{-1}\,d{\vct q}\,d{\vct k}.
\end{align*}
To derive additional symmetry properties than Lorentz covariance we will need to delve into the explicit form of the classical 3-graviton vertex $V^{(3)}$. From (Duff, 1975, p. 89) \cite{Duff} we have
\begin{align*}
&V^{(3)}_{\alpha_1\beta_1\alpha_2\beta_2\alpha_3\beta_3}(p_1,p_2,p_3)\\
=&-\text{sym P}_6\frac{1}{8}[-4\eta_{\alpha_3\alpha_2}\eta_{\beta_2\alpha_1}p_2\cdot p_3+2\eta_{\alpha_2\beta_2}\eta_{\alpha_3\alpha_1}\eta_{\beta_3\beta_1}p_2\cdot p_3-\eta_{\alpha_2\beta_2}\eta_{\alpha_3\beta_3}p_{2\alpha_1}p_{3\beta_1}\\
+&2\eta_{\alpha_3\alpha_2}\eta_{\beta_2\beta_3}p_{2\alpha_1}p_{3\beta_1}+4\eta_{\alpha_2\alpha_1}\eta_{\beta_3\beta_1}p_{2\alpha_3}p_{3\beta_2}],
\end{align*}
where ``sym" indicates a symmetrization on each index pair $\alpha_1\beta_1$,$\alpha_2\beta_2$ and $\alpha_3\beta_3$ and P$_6$ means that a summation is to be performed over all six permutations of the momentum-index triplets $\alpha_1\beta_1p_1,\alpha_2\beta_2p_2$ and $\alpha_3\beta_3p_3$.

The first symmetry property of $V^{(3)}$ clearly results in the symmetry of $\Pi^{\mu\nu\rho\sigma}$ under either the transformations $\mu\leftrightarrow\nu$ or $\rho\leftrightarrow\sigma$.

For the other symmetry property of $\Pi^{\mu\nu\rho\sigma}$ we have the following.
\begin{lemma}
\[ \Pi^{\mu\nu\rho\sigma}=\Pi^{\rho\sigma\mu\nu},\forall\mu,\nu,\rho,\sigma\in\{0,1,2,3\}. \]
\end{lemma}
{\noindent{\bf Proof}}
From the second symmetry property of $V^{(3)}$ together with the obvious property
\begin{equation}\label{eq:V_3_parity}
V^{(3)}_{\alpha_1\beta_1\alpha_2\beta_2\alpha_3\beta_3}(-p_1,-p_2,-p_3)=V^{(3)}_{\alpha_1\beta_1\alpha_2\beta_2\alpha_3\beta_3}(p_1,p_2,p_3),
\end{equation}
we have that 
\[ V^{(3)\rho\sigma}_{\alpha_1\beta_1\alpha_2\beta_2}(-q,q+k,-k)=V^{(3)\rho\sigma}_{\alpha_1\beta_1\alpha_2\beta_2}(q,-q-k,k)=V^{(3)\rho\sigma}_{\alpha_2\beta_2\alpha_1\beta_1}(-q-k,q,k), \]
(where the contravariant components of $V^{(3)}$ in this equation occupy the last two tensor index positions). Therefore, using the cyclic property of the trace and the property $P^{\mu\nu\rho\sigma}=P^{\rho\sigma\mu\nu}$, we have
\begin{align*}
&V^{(3)\mu\nu}_{\alpha_1\beta_1\alpha_2\beta_2}(-q-k,q,k)P^{\alpha_2\beta_2\alpha_3\beta_3}V^{(3)\rho\sigma}_{\alpha_3\beta_3\alpha_4\beta_4}(-q,q+k,-k)P^{\alpha_4\beta_4\alpha_1\beta_1}\\
=&V^{(3)\rho\sigma}_{\alpha_1\beta_1\alpha_2\beta_2}(-q,q+k,-k)P^{\alpha_2\beta_2\alpha_3\beta_3}V^{(3)\mu\nu}_{\alpha_3\beta_3\alpha_4\beta_4}(-q-k,q,k)P^{\alpha_4\beta_4\alpha_1\beta_1}\\
=&V^{(3)\rho\sigma}_{\alpha_2\beta_2\alpha_1\beta_1}(-q-k,q,k)P^{\alpha_2\beta_2\alpha_3\beta_3}V^{(3)\mu\nu}_{\alpha_4\beta_4\alpha_3\beta_3}(-q,q+k,-k)P^{\alpha_4\beta_4\alpha_1\beta_1}\\
=&P^{\alpha_3\beta_3\alpha_2\beta_2}V^{(3)\rho\sigma}_{\alpha_2\beta_2\alpha_1\beta_1}(-q-k,q,k)P^{\alpha_1\beta_1\alpha_4\beta_4}V^{(3)\mu\nu}_{\alpha_4\beta_4\alpha_3\beta_3}(-q,q+k,-k).
\end{align*}
Thus, using the cyclic property of the trace again, we have that the r.h.s. of the last equation in the above sequence of equations is equal to
\[ V^{(3)\rho\sigma}_{\alpha_1\beta_1\alpha_2\beta_2}(-q-k,q,k)P^{\alpha_2\beta_2\alpha_3\beta_3}V^{(3)\mu\nu}_{\alpha_3\beta_3\alpha_4\beta_4}(-q,q+k,-k)P^{\alpha_4\beta_4\alpha_1\beta_1}, \]
and the required result follows easily. $\Box$

Therefore the graviton vacuum polarization tensor satisfies the symmetry conditions given by Eq.~\ref{eq:Pi_tensor_symmetry}.

\subsection{Existence of a density for $\Pi^{\mu\nu\rho\sigma}$}

From Eq.~\ref{eq:VP_measure_formula} $\Pi^{\mu\nu\rho\sigma}$ has the form
\begin{equation}\nonumber
\Pi^{\mu\nu\rho\sigma}(\Gamma)=\int\chi_{\Gamma}(q+k)f^{\mu\nu\rho\sigma}(q,k)|{\vct q}|^{-1}|{\vct k}|^{-1}\,d{\vct q}\,d{\vct k},
\end{equation}
for some continuous $f^{\mu\nu\rho\sigma}:({\bf R}^4)^2\rightarrow{\bf R}$. Thus
\begin{align*}
\Pi^{\mu\nu\rho\sigma}(\Gamma)=&\int\chi_{\Gamma}((|{\vct q}|+|{\vct k}|,{\vct q}+{\vct k}))f(q,k)|{\vct q}|^{-1}|{\vct k}|^{-1}\,d{\vct q}\,d{\vct k}\\
=&\int_{s_1=0}^{\infty}\int_{\theta_1=0}^{\pi}\int_{\phi_1=0}^{2\pi}\int_{s_2=0}^{\infty}\int_{\theta_2=0}^{\pi}\int_{\phi_2=0}^{2\pi}\chi_{\Gamma}((s_1+s_2,s_1\omega(\theta_1,\phi_1)+s_2\omega(\theta_2,\phi_2)))\\
&s_1s_2\sin(\theta_1)\sin(\theta_2)f^{\mu\nu\rho\sigma}(q,k)\,d\phi_2\,d\theta_2\,ds_2\,d\phi_1\,d\theta_1\,ds_1,
\end{align*}
where 
\[ \omega(\theta,\phi)=(\sin(\theta)\cos(\phi),\sin(\theta)\sin(\phi),\cos(\theta)). \]
Suppose that $\Gamma\in{\mathcal B}_0({\bf R}^4)$ is relatively compact. Then there exists $R>0$ such that $s_1>R$ or $s_2>R\Rightarrow(s_1+s_2,s_1\omega(\theta_1,\phi_1)+s_2\omega(\theta_2,\phi_2)){\slas{\in}}\Gamma$. Hence
\begin{align*}
\Pi^{\mu\nu\rho\sigma}(\Gamma)=&\int_{s_1=0}^{R}\int_{\theta_1=0}^{\pi}\int_{\phi_1=0}^{2\pi}\int_{s_2=0}^{R}\int_{\theta_2=0}^{\pi}\int_{\phi_2=0}^{2\pi}\chi_{\Gamma}((s_1+s_2,s_1\omega(\theta_1,\phi_1)+s_2\omega(\theta_2,\phi_2)))\\
&s_1s_2\sin(\theta_1)\sin(\theta_2)f^{\mu\nu\rho\sigma}(q,k)\,d\phi_2\,d\theta_2\,ds_2\,d\phi_1\,d\theta_1\,ds_1,
\end{align*}
Thus, since the domain of integration is compact and the integrand is continuous, Fubini's theorem is applicable and we can interchange the order of the integrations at will.
Therefore
\[ \Pi^{\mu\nu\rho\sigma}(\Gamma)=\int_{\theta_2=0}^{\pi}\int_{\phi_2=0}^{2\pi}I^{\mu\nu\rho\sigma}_{(\theta_2,\phi_2)}(\Gamma)\sin(\theta_2)\,d\phi_2\,d\theta_2, \]
where
\begin{align*}
I_{(\theta_2,\phi_2)}^{\mu\nu\rho\sigma}(\Gamma)=&\int_{s_1=0}^{R}\int_{s_2=0}^{R}\int_{\theta_1=0}^{\pi}\int_{\phi_1=0}^{2\pi}\chi_{\Gamma}((s_1+s_2,s_1\omega(\theta_1,\phi_1)+s_2\omega(\theta_2,\phi_2)))\\
&s_1s_2\sin(\theta_1)f^{\mu\nu\rho\sigma}(q,k)\,d\phi_1\,d\theta_1\,ds_1\,ds_2,
\end{align*}
Make a change of variables transformation given by
\[ p=p_{(\theta_2,\phi_2)}(s_1,s_2,\theta_1,\phi_1)=h_{(\theta_2,\phi_2)}(s_1,s_2,\theta_1,\phi_1)=(s_1+s_2,s_1\omega(\theta_1,\phi_1)+s_2\omega(\theta_2,\phi_2)). \]
The Jacobian of the transformation for any $(\theta_2,\phi_2)\in[0,\pi]\times[0,2\pi]$ is 
\begin{align*}
&J_{(\theta_2,\phi_2)}(s_1,s_2,\theta_1,\phi_1)=\\
&\left|\text{det}\left(\begin{array}{cccc}
1&1&0&0\\
\sin(\theta_1)\cos(\phi_1)&\sin(\theta_2)\cos(\phi_2)&s_1\cos(\theta_1)\cos(\phi_1)&-s_1\sin(\theta_1)\sin(\phi_1)\\
\sin(\theta_1)\sin(\phi_1)&\sin(\theta_2)\sin(\phi_2)&s_1\cos(\theta_1)\sin(\phi_1)&s_1\sin(\theta_1)\cos(\phi_1)\\
\cos(\theta_1)&\cos(\theta_2)&-s_1\sin(\theta_1)&0
\end{array}\right)\right|.
\end{align*}
The determinant is an analytic function of $(s_1,s_2,\theta_1,\phi_1)$ and it vanishes on a submanifold of $(0,\infty)^2\times[0,\pi]\times[0,2\pi]$ of co-dimension 1 (or more) which, therefore, has 4-Lebesgue measure zero. (The inverse image of 0 by det is a closed submanifold of $(0,\infty)^2\times[0,\pi]\times[0,2\pi]$ which, by analyticity, must have empty interior. Thus, by the Invariance of Domain theorem it must have topological dimension $<4$). Therefore
\begin{align*}
I_{(\theta_2,\phi_2)}^{\mu\nu\rho\sigma}(\Gamma)=&\int\chi_{\Gamma}(p)s_1s_2\sin(\theta_1)f^{\mu\nu\rho\sigma}(q,k)(J_{(\theta_2,\phi_2)}(h_{(\theta_2,\phi_2)}^{-1}(p)))^{-1}\,dp,
\end{align*}
Hence ${\mathcal L}^4(\Gamma)=0\Rightarrow I_{(\theta_2,\phi_2)}(\Gamma)=0,\forall(\theta_2,\phi_2)\in[0,\pi]\times[0,2\pi]$ where ${\mathcal L}^4$ denotes the Lebesgue measure on ${\bf R}^4$.
Hence ${\mathcal L}^4(\Gamma)=0\Rightarrow\Pi^{\mu\nu\rho\sigma}(\Gamma)=0$. Therefore by the Radon-Nikodym theorem $\Pi^{\mu\nu\rho\sigma}$ is absolutely continuous with respect to ${\mathcal L}^4$ i.e. there exists a density $\Pi^{\mu\nu\rho\sigma}:{\bf R}^4\rightarrow{\bf C}$ such that
\[ \Pi^{\mu\nu\rho\sigma}(\Gamma)=\int_{\Gamma}\Pi^{\mu\nu\rho\sigma}(q)\,dq,\forall\Gamma\in{\mathcal B}_0({\bf R}^4). \]

If the spectral functions are continuous then the density $\Pi$ is continuous and if they are analytic then the density is analytic.

\subsection{Computation of the spectral functions for $\Pi^{\mu\nu\rho\sigma}$ using the spectral calculus}

We now compute the functions $g^{\mu\nu\rho\sigma\prime}_a$ associated with the measure $\Pi^{\mu\nu\rho\sigma}$ so that we can use the results of Section~\ref{section:spectral_calculus_summary} to compute the spectral functions $\sigma_i$ and hence the density $\Pi^{\mu\nu\rho\sigma}$ associated with the measure.
Using Eq.~\ref{eq:VP_measure_formula} we make the following rigorous computation
\begin{align*}
g^{\mu\nu\rho\sigma}(a,b,\epsilon)=&\Pi^{\mu\nu\rho\sigma}(\Gamma(a,b,\epsilon))\\
=&-\frac{\kappa^2}{16\pi^2}\int\chi_{\Gamma(a,b,\epsilon)}(q+k)V^{(3)\mu\nu}_{\alpha_1\beta_1\alpha_2\beta_2}(-q-k,q,k)P^{\alpha_2\beta_2\alpha_3\beta_3}\\
&V^{(3)\rho\sigma}_{\alpha_3\beta_3\alpha_4\beta_4}(-q,q+k,-k)P^{\alpha_4\beta_4\alpha_1\beta_1}|{\vct q}|^{-1}|{\vct k}|^{-1}\,d{\vct q}\,d{\vct k}\\
\approx&-\frac{\kappa^2}{16\pi^2}\int\chi_{(a,b)}(|{\vct q}|+|{\vct k}|)\chi_{B_{\epsilon}({\vct0})}({\vct q}+{\vct k})V^{(3)\mu\nu}_{\alpha_1\beta_1\alpha_2\beta_2}(-q-k,q,k)P^{\alpha_2\beta_2\alpha_3\beta_3}\\
&V^{(3)\rho\sigma}_{\alpha_3\beta_3\alpha_4\beta_4}(-q,q+k,-k)P^{\alpha_4\beta_4\alpha_1\beta_1}|{\vct q}|^{-1}|{\vct k}|^{-1}\,d{\vct q}\,d{\vct k}\\
=&-\frac{\kappa^2}{16\pi^2}\int\chi_{(a,b)}(|{\vct q}|+|{\vct k}|)\chi_{B_{\epsilon}({\vct0})-{\vct k}}({\vct q})V^{(3)\mu\nu}_{\alpha_1\beta_1\alpha_2\beta_2}(-q-k,q,k)P^{\alpha_2\beta_2\alpha_3\beta_3}\\
&V^{(3)\rho\sigma}_{\alpha_3\beta_3\alpha_4\beta_4}(-q,q+k,-k)P^{\alpha_4\beta_4\alpha_1\beta_1}|{\vct q}|^{-1}|{\vct k}|^{-1}\,d{\vct q}\,d{\vct k}\\
\approx&-\frac{\kappa^2}{16\pi^2}(\frac{4}{3}\pi\epsilon^3)\int\chi_{(a,b)}(2|{\vct k}|)V^{(3)\mu\nu}_{\alpha_1\beta_1\alpha_2\beta_2}(-q-k,q,k)P^{\alpha_2\beta_2\alpha_3\beta_3}\\
&V^{(3)\rho\sigma}_{\alpha_3\beta_3\alpha_4\beta_4}(-q,q+k,-k)P^{\alpha_4\beta_4\alpha_1\beta_1}|{\vct k}|^{-2}\,d{\vct k},
\end{align*}
where $q=q({\vct k})=(|{\vct k}|,-{\vct k}),k=k({\vct k})=(|{\vct k}|,{\vct k})$. Now
\[ \chi_{(a,b)}(2|{\vct k}|)=1\Leftrightarrow a<2|{\vct k}|<b\Leftrightarrow\frac{a}{2}<|{\vct k}|<\frac{b}{2}. \]
Therefore
\begin{align*}
g^{\mu\nu\rho\sigma}(a,b,\epsilon)\approx&-\frac{\kappa^2}{16\pi^2}(\frac{4}{3}\pi\epsilon^3)\int_{s=\frac{a}{2}}^\frac{b}{2}\int_{\theta=0}^{\pi}\int_{\phi=0}^{2\pi}V^{(3)\mu\nu}_{\alpha_1\beta_1\alpha_2\beta_2}(-q-k,q,k)P^{\alpha_2\beta_2\alpha_3\beta_3}\\
&V^{(3)\rho\sigma}_{\alpha_3\beta_3\alpha_4\beta_4}(-q,q+k,-k)P^{\alpha_4\beta_4\alpha_1\beta_1}s^{-2}s^2\sin(\theta)\,d\phi\,d\theta\,ds\\
=&-\frac{\kappa^2}{12\pi}\epsilon^3\int_{s=\frac{a}{2}}^\frac{b}{2}\int_{\theta=0}^{\pi}\int_{\phi=0}^{2\pi}V^{(3)\mu\nu}_{\alpha_1\beta_1\alpha_2\beta_2}(-q-k,q,k)P^{\alpha_2\beta_2\alpha_3\beta_3}\\
&V^{(3)\rho\sigma}_{\alpha_3\beta_3\alpha_4\beta_4}(-q,q+k,-k)P^{\alpha_4\beta_4\alpha_1\beta_1}\sin(\theta)\,d\phi\,d\theta\,ds,
\end{align*}
where $q=q({\vct k}),k=k({\vct k})$ and ${\vct k}={\vct k}(s,\theta,\phi)=s(\sin(\theta)\cos(\phi),\sin(\theta)\sin(\phi),\cos(\theta))$. Hence
\begin{align*}
g^{\mu\nu\rho\sigma}_a(b)=&\lim_{\epsilon\rightarrow0}(\epsilon^{-3}g(a,b,\epsilon))\\
=&-\frac{\kappa^2}{12\pi}\int_{s=\frac{a}{2}}^\frac{b}{2}\int_{\theta=0}^{\pi}\int_{\phi=0}^{2\pi}V^{(3)\mu\nu}_{\alpha_1\beta_1\alpha_2\beta_2}(-q-k,q,k)P^{\alpha_2\beta_2\alpha_3\beta_3}\\
&V^{(3)\rho\sigma}_{\alpha_3\beta_3\alpha_4\beta_4}(-q,q+k,-k)P^{\alpha_4\beta_4\alpha_1\beta_1}\sin(\theta)\,d\phi\,d\theta\,ds.
\end{align*}
Therefore, using the Leibniz integral rule, we have
\begin{align}\label{eq:g_prime_functions}
g^{\mu\nu\rho\sigma\prime}_a(b)=&-\frac{\kappa^2}{24\pi}\int_{\theta=0}^{\pi}\int_{\phi=0}^{2\pi}V^{(3)\mu\nu}_{\alpha_1\beta_1\alpha_2\beta_2}(-q-k,q,k)P^{\alpha_2\beta_2\alpha_3\beta_3}\nonumber\\
&V^{(3)\rho\sigma}_{\alpha_3\beta_3\alpha_4\beta_4}(-q,q+k,-k)P^{\alpha_4\beta_4\alpha_1\beta_1}\sin(\theta)\,d\phi\,d\theta,
\end{align}
where $q=q({\vct k}),k=k({\vct k})$ and ${\vct k}={\vct k}(b,\theta,\phi)=b(\sin(\theta)\cos(\phi),\sin(\theta)\sin(\phi),\cos(\theta))$.
The values $g^{\mu\nu\rho\sigma\prime}_a(b)$ are well defined and finite for all $b>0$ since the 3-graviton vertex functions $V^{(3)}$ are continuous and $[0,\pi]\times[0,2\pi]$ is compact.

Since the graviton vertex function $V^{(3)}$ is polynomial it is analytic. Also $[0,\pi]\times[0,2\pi]$ is compact. Therefore, by the well known theorem in complex analysis concerning the analyticity of parameter-dependent integrals (differentiation under the integral sign), the $g_a^{\mu\nu\rho\sigma\prime}$ functions are analytic. 

Therefore using the spectral calculus algorithm described in Section~\ref{section:spectral_calculus_summary} the spectral functions $\sigma_1,\ldots,\sigma_{10}$ can be computed and they are analytic.

\subsection{The graviton vacuum polarization function $\Pi^{\mu\nu\rho\sigma}$}

The density associated with the analytic spectral functions as defined by Eq.~\ref{eq:canonical_form_density} is analytic.
It is a well defined and finite analytic complex tensor valued function. 

We have computed it from a causal measure associated with the vacuum polarization Feynman integral. Therefore the density that we have computed is defined in the timelike region of Minkowski space. We would like to use it to analyze graviton scattering processes for which the momentum transfer is spacelike. Therefore we would like to analytically continue the vacuum polarization function that we have computed to the spacelike region of Minkowski space.

This is straightforward because the graviton 3-vertex is a polynomial function so it is defined and analytic over all of ${\bf R}^4$ and the $P^{\mu\nu\rho\sigma}$ are just numbers. Thus, by Eq.~\ref{eq:g_prime_functions}, the functions $g_a^{\mu\nu\rho\sigma\prime}(Q)$ are defined and analytic for all $Q\in{\bf R}$. Following Ref.~\cite{NPB} we define $\zeta(q)=\text{sign}(q^2)|q^2|^{\frac{1}{2}}$ ($\zeta$ generates a one-to-one mapping between ${\bf R}\backslash\{0\}$ and the non-null orbits of $O(1,3)^{\uparrow}$ in Minkowski space) and extend $\Pi^{\mu\nu\rho\sigma}$ to all of Minkowski space using the formula
\begin{align*}
\Pi^{\mu\nu\rho\sigma}(q)=&Q(\eta^{\mu\nu}\eta^{\rho\sigma}\sigma_1(Q)+\eta^{\mu\rho}\eta^{\nu\sigma}\sigma_2(Q)+\eta^{\mu\sigma}\eta^{\nu\rho}\sigma_3(Q))+\\
&Q^{-1}(\eta^{\mu\nu}q^{\rho}q^{\sigma}\sigma_4(Q)+\eta^{\mu\rho}q^{\nu}q^{\sigma}\sigma_5(Q)+\eta^{\mu\sigma}q^{\nu}q^{\rho}\sigma_6(Q)+\\
&\eta^{\nu\rho}q^{\mu}q^{\sigma}\sigma_7(Q)+\eta^{\nu\sigma}q^{\mu}q^{\rho}\sigma_8(Q)+\eta^{\rho\sigma}q^{\mu}q^{\nu}\sigma_9(Q)+\\
&q^{\mu}q^{\nu}q^{\rho}q^{\sigma}\sigma_{10}(Q)),
\end{align*}
where $Q=\zeta(q)$. In particular, when $q$ is spacelike, i.e. $q^2<0$, then $\zeta(q)=-(-q^2)^{\frac{1}{2}}$. Thus $\Pi^{\mu\nu\rho\sigma}$ is defined for spacelike $q$ by
\begin{align}\label{eq:spacelike_graviton_VP}
\Pi^{\mu\nu\rho\sigma}(q)=&-Q(\eta^{\mu\nu}\eta^{\rho\sigma}\sigma_1(-Q)+\eta^{\mu\rho}\eta^{\nu\sigma}\sigma_2(-Q)+\eta^{\mu\sigma}\eta^{\nu\rho}\sigma_3(-Q))-\nonumber\\
&Q^{-1}(\eta^{\mu\nu}q^{\rho}q^{\sigma}\sigma_4(-Q)+\eta^{\mu\rho}q^{\nu}q^{\sigma}\sigma_5(-Q)+\eta^{\mu\sigma}q^{\nu}q^{\rho}\sigma_6(-Q)+\nonumber\\
&\eta^{\nu\rho}q^{\mu}q^{\sigma}\sigma_7(-Q)+\eta^{\nu\sigma}q^{\mu}q^{\rho}\sigma_8(-Q)+\eta^{\rho\sigma}q^{\mu}q^{\nu}\sigma_9(-Q)+\nonumber\\
&q^{\mu}q^{\nu}q^{\rho}q^{\sigma}\sigma_{10}(-Q)),
\end{align}
where $Q=(-q^2)^{\frac{1}{2}}$.

This function can be used in QG calculations involving the vacuum polarization tensor for scattering processes in the t-channel without the need for counterterms.

\section{Graviton-graviton scattering in pure gravity QG}

The following gives some examples of multi-loop QG calculations using the graviton vacuum polarization (VP) tensor. These calculations involve a number of insertions of a VP bubble at various points along a graviton propagator.

\subsection{Tree level}

A Feynman diagram associated with tree level graviton-graviton scattering is shown in Fig. \ref{fig:tree-level}. The Feynman amplitude for this process is is given by
\begin{align*}
i{\mathcal M}_{\alpha_1^{\prime}\alpha_2^{\prime}\alpha_1\alpha_2}(p_1^{\prime},p_2^{\prime},p_1,p_2)=&\epsilon^{\mu_1^{\prime}\nu_1^{\prime}}(p_1^{\prime},\alpha_1^{\prime})^{*}i\kappa V^{(3)\rho\sigma}_{\mu_1^{\prime}\nu_1^{\prime}\mu_1\nu_1}(p_1^{\prime},q,-p_1)\epsilon^{\mu_1\nu_1}(p_1,\alpha_1)\\
&iD_{\rho\sigma\gamma\delta}(q)\\
&\epsilon^{\mu_2^{\prime}\nu_2^{\prime}}(p_2^{\prime},\alpha_2^{\prime})^{*}i\kappa V^{(3)\gamma\delta}_{\mu_2^{\prime}\nu_2^{\prime}\mu_2\nu_2}(-q,p_2^{\prime},-p_2)\epsilon^{\mu_2\nu_2}(p_2,\alpha_2), 
\end{align*}
where $q=p_1-p_1^{\prime}$ is the momentum transfer, $\{\epsilon^{\mu\nu}(p,\alpha):\alpha\in\{1,2\},p\in{\bf R}^4,\mu,\nu\in\{0,1,2,3\}\}$ are polarization tensors which satisfy \cite{Misner}
\begin{align*}
&p_{\mu}\epsilon^{\mu\nu}(p,\alpha)=0\text{ (transversality)},\\
&\eta_{\mu\nu}\epsilon^{\mu\nu}(p,\alpha)=0\text{ (trace free property)},\\
&\epsilon^{\mu\nu}(p,\alpha)=\epsilon^{\nu\mu}(p,\alpha)\text{ (symmetry)},
\end{align*} 
and $\kappa=\sqrt{32\pi G}$ is the gravitational coupling constant in which $G=G_N$ is Newton's universal constant of gravitation.

\begin{figure} 
\centering
\includegraphics[width=8cm]{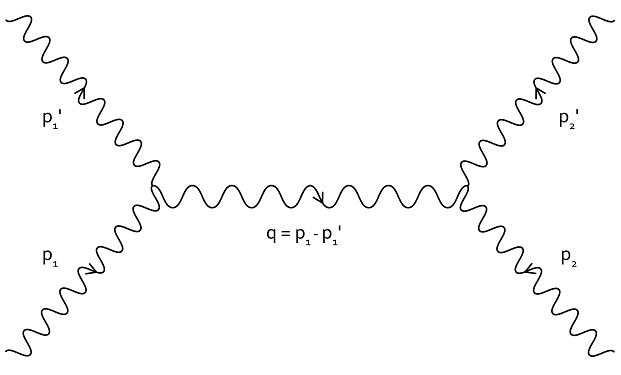}
\caption{A Feynman diagram for tree-level graviton-graviton scattering\label{fig:tree-level}}
\end{figure}

Therefore
\begin{align*}
{\mathcal M}_{\alpha_1^{\prime}\alpha_2^{\prime}\alpha_1\alpha_2}(p_1^{\prime},p_2^{\prime},p_1,p_2)=&\frac{\kappa^2}{Q^2}\epsilon^{\mu_1^{\prime}\nu_1^{\prime}}(p_1^{\prime},\alpha_1^{\prime})^{*}V^{(3)\rho\sigma}_{\mu_1^{\prime}\nu_1^{\prime}\mu_1\nu_1}(p_1^{\prime},q,-p_1)\epsilon^{\mu_1\nu_1}(p_1,\alpha_1)\\
&P_{\rho\sigma\gamma\delta}(q)\\
&\epsilon^{\mu_2^{\prime}\nu_2^{\prime}}(p_2^{\prime},\alpha_2^{\prime})^{*}V^{(3)\gamma\delta}_{\mu_2^{\prime}\nu_2^{\prime}\mu_2\nu_2}(-q,p_2^{\prime},-p_2)\epsilon^{\mu_2\nu_2}(p_2,\alpha_2), 
\end{align*}
where $Q=(-q^2)^{\frac{1}{2}}$.

\subsection{One-loop example}

The one-loop Feynman amplitude associated with the Feynman diagram of Fig. \ref{fig:one-loop} is
\begin{align*}
i{\mathcal M}_{\alpha_1^{\prime}\alpha_2^{\prime}\alpha_1\alpha_2}(p_1^{\prime},p_2^{\prime},p_1,p_2)=&\epsilon^{\mu_1^{\prime}\nu_1^{\prime}}(p_1^{\prime},\alpha_1^{\prime})^{*}i\kappa V^{(3)\rho\sigma}_{\mu_1^{\prime}\nu_1^{\prime}\mu_1\nu_1}(p_1^{\prime},q,-p_1)\epsilon^{\mu_1\nu_1}(p_1,\alpha_1)\\
&iD_{\rho\sigma\rho_1\sigma_1}(q)i\Pi^{\rho_1\sigma_1\gamma_1\delta_1}(q)iD_{\gamma_1\delta_1\gamma\delta}(q)\\
&\epsilon^{\mu_2^{\prime}\nu_2^{\prime}}(p_2^{\prime},\alpha_2^{\prime})^{*}i\kappa V^{(3)\gamma\delta}_{\mu_2^{\prime}\nu_2^{\prime}\mu_2\nu_2}(-q,p_2^{\prime},-p_2)\epsilon^{\mu_2\nu_2}(p_2,\alpha_2).   
\end{align*}

\begin{figure} 
\centering
\includegraphics[width=8cm]{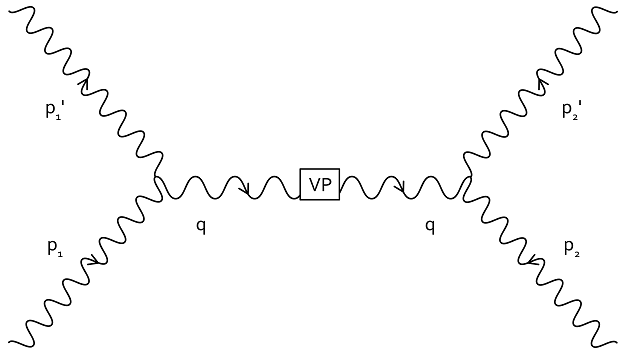}
\caption{A Feynman diagram for one-loop graviton-graviton scattering\label{fig:one-loop}}
\end{figure}

Therefore
\begin{align*}
{\mathcal M}_{\alpha_1^{\prime}\alpha_2^{\prime}\alpha_1\alpha_2}(p_1^{\prime},p_2^{\prime},p_1,p_2)=&\frac{\kappa^2}{Q^4}\epsilon^{\mu_1^{\prime}\nu_1^{\prime}}(p_1^{\prime},\alpha_1^{\prime})^{*}V^{(3)\rho\sigma}_{\mu_1^{\prime}\nu_1^{\prime}\mu_1\nu_1}(p_1^{\prime},q,-p_1)\epsilon^{\mu_1\nu_1}(p_1,\alpha_1)\\
&P_{\rho\sigma\rho_1\sigma_1}\Pi^{\rho_1\sigma_1\gamma_1\delta_1}(q)P_{\gamma_1\delta_1\gamma\delta}\\
&\epsilon^{\mu_2^{\prime}\nu_2^{\prime}}(p_2^{\prime},\alpha_2^{\prime})^{*}V^{(3)\gamma\delta}_{\mu_2^{\prime}\nu_2^{\prime}\mu_2\nu_2}(-q,p_2^{\prime},-p_2)\epsilon^{\mu_2\nu_2}(p_2,\alpha_2).
\end{align*}
It is straightforward to show that the polarization tensors can be chosen to be analytic functions of their momentum argument. As mentioned above the 3-graviton vertex and the VP function are analytic. Therefore the Feynman amplitude is an analytic function of its momentum arguments. 

It is common, e.g. see Refs.~\cite{Feynman_ghosts,tHooft} in Yang-Mills and QG theory to include in the analysis ghost particles which are fictitious particles circulating in loops which are supposedly required to ensure unitary of the S matrix associated with the Feynman amplitude ${\mathcal M}$.

If one wants to include such fictitious particles then the required computations can easily be carried out using CSR. However we have found that when using CSR loops respect unitarity and ghost particles are not necessary. This result will be presented elsewhere. 

\subsection{Two-loop example}

The two-loop Feynman amplitude associated with the Feynman diagram of Fig. \ref{fig:two-loop} is
\begin{align*}
i{\mathcal M}_{\alpha_1^{\prime}\alpha_2^{\prime}\alpha_1\alpha_2}(p_1^{\prime},p_2^{\prime},p_1,p_2)=&\epsilon^{\mu_1^{\prime}\nu_1^{\prime}}(p_1^{\prime},\alpha_1^{\prime})^{*}i\kappa V^{(3)\rho\sigma}_{\mu_1^{\prime}\nu_1^{\prime}\mu_1\nu_1}(p_1^{\prime},q,-p_1)\epsilon^{\mu_1\nu_1}(p_1,\alpha_1)\\
&iD_{\rho\sigma\rho_1\sigma_1}(q)i\Pi^{\rho_1\sigma_1\rho_2\sigma_2}(q)iD_{\rho_2\sigma_2\rho_3\sigma_3}(q)i\Pi^{\rho_3\sigma_3\rho_4\sigma_4}(q)iD_{\rho_4\sigma_4\gamma\delta}(q)\\
&\epsilon^{\mu_2^{\prime}\nu_2^{\prime}}(p_2^{\prime},\alpha_2^{\prime})^{*}i\kappa V^{(3)\gamma\delta}_{\mu_2^{\prime}\nu_2^{\prime}\mu_2\nu_2}(-q,p_2^{\prime},-p_2)\epsilon^{\mu_2\nu_2}(p_2,\alpha_2).
\end{align*}

\begin{figure} 
\centering
\includegraphics[width=8cm]{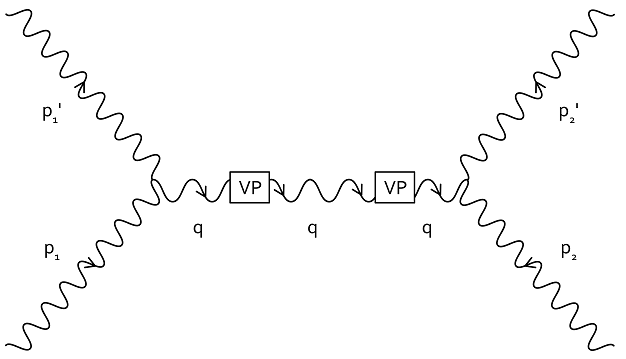}
\caption{A Feynman diagram for two-loop graviton-graviton scattering\label{fig:two-loop}}
\end{figure}

Therefore
\begin{align*}
{\mathcal M}_{\alpha_1^{\prime}\alpha_2^{\prime}\alpha_1\alpha_2}(p_1^{\prime},p_2^{\prime},p_1,p_2)=&\frac{\kappa^2}{Q^6}\epsilon^{\mu_1^{\prime}\nu_1^{\prime}}(p_1^{\prime},\alpha_1^{\prime})^{*}V^{(3)\rho\sigma}_{\mu_1^{\prime}\nu_1^{\prime}\mu_1\nu_1}(p_1^{\prime},q,-p_1)\epsilon^{\mu_1\nu_1}(p_1,\alpha_1)\\
&P_{\rho\sigma\rho_1\sigma_1}\Pi^{\rho_1\sigma_1\rho_2\sigma_2}(q)P_{\rho_2\sigma_2\rho_3\sigma_3}\Pi^{\rho_3\sigma_3\rho_4\sigma_4}(q)P_{\rho_4\sigma_4\gamma\delta}\\
&\epsilon^{\mu_2^{\prime}\nu_2^{\prime}}(p_2^{\prime},\alpha_2^{\prime})^{*}V^{(3)\gamma\delta}_{\mu_2^{\prime}\nu_2^{\prime}\mu_2\nu_2}(-q,p_2^{\prime},-p_2)\epsilon^{\mu_2\nu_2}(p_2,\alpha_2).
\end{align*}

\subsection{$n$-loop example}

The $n$-loop Feynman amplitude associated with the Feynman diagram of Fig. \ref{fig:n-loop} is
\begin{align*}
i{\mathcal M}_{\alpha_1^{\prime}\alpha_2^{\prime}\alpha_1\alpha_2}(p_1^{\prime},p_2^{\prime},p_1,p_2)=&\epsilon^{\mu_1^{\prime}\nu_1^{\prime}}(p_1^{\prime},\alpha_1^{\prime})^{*}i\kappa V^{(3)\rho_0\sigma_0}_{\mu_1^{\prime}\nu_1^{\prime}\mu_1\nu_1}(p_1^{\prime},q,-p_1)\epsilon^{\mu_1\nu_1}(p_1,\alpha_1)\\
&\prod_{k=0}^{n-1}(iD_{\rho_{2k}\sigma_{2k}\rho_{2k+1}\sigma_{2k+1}}(q)i\Pi^{\rho_{2k+1}\sigma_{2k+1}\rho_{2k+2}\sigma_{2k+2}}(q))iD_{\rho_{2n}\sigma_{2n}\gamma\delta}(q)\\
&\epsilon^{\mu_2^{\prime}\nu_2^{\prime}}(p_2^{\prime},\alpha_2^{\prime})^{*}i\kappa V^{(3)\gamma\delta}_{\mu_2^{\prime}\nu_2^{\prime}\mu_2\nu_2}(-q,p_2^{\prime},-p_2)\epsilon^{\mu_2\nu_2}(p_2,\alpha_2).
\end{align*}

\begin{figure} 
\centering
\includegraphics[width=8cm]{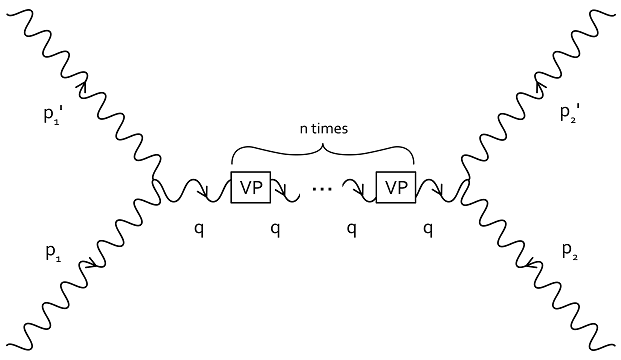}
\caption{A Feynman diagram for $n$-loop graviton-graviton scattering\label{fig:n-loop}}
\end{figure}

Therefore
\begin{align*}
{\mathcal M}_{\alpha_1^{\prime}\alpha_2^{\prime}\alpha_1\alpha_2}(p_1^{\prime},p_2^{\prime},p_1,p_2)=&\frac{\kappa^2}{Q^{2(n+1)}}\epsilon^{\mu_1^{\prime}\nu_1^{\prime}}(p_1^{\prime},\alpha_1^{\prime})^{*}V^{(3)\rho_0\sigma_0}_{\mu_1^{\prime}\nu_1^{\prime}\mu_1\nu_1}(p_1^{\prime},q,-p_1)\epsilon^{\mu_1\nu_1}(p_1,\alpha_1)\\
&\prod_{k=0}^{n-1}(P_{\rho_{2k}\sigma_{2k}\rho_{2k+1}\sigma_{2k+1}}\Pi^{\rho_{2k+1}\sigma_{2k+1}\rho_{2k+2}\sigma_{2k+2}}(q))P_{\rho_{2n}\sigma_{2n}\gamma\delta}\\
&\epsilon^{\mu_2^{\prime}\nu_2^{\prime}}(p_2^{\prime},\alpha_2^{\prime})^{*}V^{(3)\gamma\delta}_{\mu_2^{\prime}\nu_2^{\prime}\mu_2\nu_2}(-q,p_2^{\prime},-p_2)\epsilon^{\mu_2\nu_2}(p_2,\alpha_2).
\end{align*}
$V^{(3)\rho\sigma}_{\mu^{\prime}\nu^{\prime}\mu\nu}$ is an analytic function of its momentum arguments, the polarization functions can be chosen to be analytic functions of their momentum arguments, the $P_{\rho\sigma\gamma\delta}$ are just numbers and the $\Pi^{\mu\nu\rho\sigma}$ are analytic functions of (the projection onto ${\bf R}^3$ of) their momentum arguments. Therefore the Feynman amplitude ${\mathcal M}$ is a well defined, finite, analytic function of its momentum arguments on the set
\begin{align*}
S=&\{(p_1^{\prime},p_2^{\prime},p_1,p_2)\in({\bf R}^4)^4:p_1^{\prime}+p_2^{\prime}=p_1+p_2,p_1^{\prime2}=p_2^{\prime2}=p_1^2=p_2^2=0,(p_1-p_1^{\prime})^2\neq0\}\\
=&\{(p_1^{\prime},p_2^{\prime},p_1,p_2)\in({\bf R}^4)^4:p_1^{\prime}+p_2^{\prime}=p_1+p_2,p_1^{\prime2}=p_2^{\prime2}=p_1^2=p_2^2=0,{\vct p}_1^{\prime}\neq{\vct p}_1\}.
\end{align*}
It is straightforward to compute the Feynman amplitude associated with multiple nested loop Feynman diagrams using CSR. In order to keep this paper to a reasonable length we will present such a calculation elsewhere.

\section{Conclusion}

Locally conformally flat spacetimes (M\"{o}bius structures) represent the most elementary family of topological structures for which the notion of Lorentz covariance is applicable.

The  subject of radiative corrections in QG is closely tied up with the subject of Lorentz covariant complex tensor valued functions on Minkowski space which is itself closely related to the subject of Lorentz covariant complex tensor valued measures on Minkowski space.

This paper has investigated the general theory of such measures in the case of four Lorentz indices. It has been shown how CSR can be used to render recalcitrant Feynman integrals as being Lorentz covariant Borel complex tensor valued measures and how spectral calculus can be used to compute the spectra of such measures and hence the densities which generate them. These densities are well defined and finite and can be used in place of the function associated with the Feynman integral in QFT and QG computations. 

The process of CSR does not involve the use of renormalization or require counterterms to cancel infinities since no infinities, or divergences, of either a UV or an IR nature arise when carrying out CSR.

These techniques have been used to compute the graviton vacuum polarization tensor which has then been used to compute the Feynman amplitude for multiloop graviton-graviton scattering processes in pure gravity QG.

In subsequent work we will investigate the incorporation of SM matter through QED \cite{IJMPA,NPB}, electroweak or QCD \cite{IJMPA_2025} into QG. This will involve equipping Feynman amplitude kernels ${\mathcal M}$ with additional structure in the form of additional indices and transformation properties. It is expected that, using CSR, a finite theory, as in pure gravity QG, will result without the impediment of an uncontrollable proliferation of divergences and counterterms.

\end{document}